\documentclass[a4paper,10pt]{article}
\usepackage{amssymb,amsmath,amsthm,mathrsfs}
\usepackage{fullpage}
\usepackage{hyperref}
\usepackage{eucal}
\usepackage{graphicx}
\usepackage{stackrel}
\usepackage{caption}
\DeclareCaptionLabelFormat{cont}{#1~#2(\alph{ContinuedFloat})}
\newcommand{\ie}{\emph{i.e.}}
\newcommand{\eg}{\emph{e.g.}}
\newcommand{\cf}{\emph{cf.}}
\newcommand{\Real}{\mathbb{R}}

\newcommand{\Nat}{\mathbb{N}}
\newcommand{\Int}{\mathbb{Z}}

\newcommand{\Dom}{\mathsf{D}}

\newcommand{\diag}{\mathop{\mathrm{diag}}\nolimits}

\newcommand{\sii}{L^2}

\newcommand{\der}{\mathrm{d}}
\newcommand{\ce}{\mathop{\mathrm{ce}}\nolimits}
\newcommand{\se}{\mathop{\mathrm{se}}\nolimits}
\makeatletter
\newcommand{\subalign}[1]{%
  \vcenter{%
    \Let@ \restore@math@cr \default@tag
    \baselineskip\fontdimen10 \scriptfont\tw@
    \advance\baselineskip\fontdimen12 \scriptfont\tw@
    \lineskip\thr@@\fontdimen8 \scriptfont\thr@@
    \lineskiplimit\lineskip
    \ialign{\hfil$\m@th\scriptstyle##$&$\m@th\scriptstyle{}##$\crcr
      #1\crcr
    }%
  }
}
\makeatother
\newtheorem{Theorem}{Theorem}

\newtheorem{Proposition}{Proposition}
\newtheorem{Corollary}{Corollary}

\theoremstyle{definition}
\newtheorem{Remark}{Remark}

\usepackage[normalem]{ulem}
\usepackage{color}
\definecolor{DarkBlue}{rgb}{0,0.1,0.7}
\definecolor{DarkRed}{rgb}{0.8,0.1,0.1}

\newcommand\soutD{\bgroup\markoverwith
{\textcolor{DarkBlue}{\rule[.01ex]{2pt}{1pt}}}\ULon}
\newcommand{\Hm}[1]{\leavevmode{\marginpar{\tiny%
$\hbox to 0mm{\hspace*{-0.5mm}$\leftarrow$\hss}%
\vcenter{\vrule depth 0.1mm height 0.1mm width \the\marginparwidth}%
\hbox to
0mm{\hss$\rightarrow$\hspace*{-0.5mm}}$\\\relax\raggedright #1}}}

\begin{document}
%
%-------%
% TITLE %
%-------%
%------------------------------------------%
%------------------------------------------%
\title{\textbf{\LARGE Effective quantum dynamics on the M\"obius strip}}
\author{Tom\'a\v{s} Kalvoda,$^a$ \
David Krej\v{c}i\v{r}{\'\i}k\,$^b$
\ and \ Kate\v{r}ina Zahradov\'a\,$^c$}
\date{\small 
\begin{quote}
\emph{
\begin{itemize}
\item[$a)$]  
Department of Applied Mathematics, 
Faculty of Information Technology,
Czech Technical University in Prague,
Th\'akurova 9, 16000 Prague 6, Czechia; 
tomas.kalvoda@fit.cvut.cz.%
\item[$b)$] 
Department of Mathematics, Faculty of Nuclear Sciences and 
Physical Engineering, Czech Technical University in Prague, 
Trojanova 13, 12000 Prague 2, Czechia;
david.krejcirik@fjfi.cvut.cz.%
\item[$c)$] 
School of Mathematical Sciences, Queen Mary University of London, London E  4NS, United Kingdom;
k.zahradova@qmul.ac.uk.%
\end{itemize}
}
\end{quote}
\today}
\maketitle
\begin{abstract}
\noindent
The Laplace--Beltrami operator in the curved M\"obius strip
is investigated in the limit when the width of the strip tends to zero.
By establishing a norm-resolvent convergence,
it is shown that spectral properties of the operator are approximated well
by an unconventional flat model whose spectrum can be computed explicitly
in terms of Mathieu functions.
Contrary to the traditional flat M\"obius strip,
our effective model contains a geometric potential. 
A comparison of the three models is made
and analytical results are accompanied by numerical computations.
\bigskip
\begin{itemize}
\item[\textbf{Keywords:}] Möbius strip, Laplace--Beltrami operator, spectrum, effective Hamiltonian, resolvent convergence,
quantisation on submanifolds.
%\item[\textbf{MSC (2010):}]
\end{itemize}
\end{abstract}
%
%------------------------------------------%
%------------------------------------------%
 
%---------------------%
\section{Introduction}
%---------------------%
%
The unorientable nature of the M\"obius strip 
has fascinated scientists as well as laypeople
since its discovery in the nineteenth century
(see~\cite{Pickover} for a popular overview),
or perhaps even before (\cf~\cite{Cartwright-Gonzalez_2016}).
In this paper we are interested in the interplay between
the peculiar geometry of the M\"obius strip 
and its physical properties quantified by spectral data,
which seems to have escaped the attention of the scientific community so far.  

\paragraph{The true model.}
We start with the traditional geometric realisation
of the M\"obius strip as a two-dimensional ruled surface
built along a circle of radius $R>0$ in~$\Real^3$:
\begin{multline}
\label{true}
  \Omega := \Big\{
  \left( 
  \left[R-t \cos\left(\frac{s}{2R}\right)\right] 
  \cos\left(\frac{s}{R}\right), 
  \left[R-t \cos\left(\frac{s}{2R}\right)\right] 
  \sin\left(\frac{s}{R}\right), 
  - t \sin\left(\frac{s}{2 R}\right)
  \right)
  \\
  : \ s \in [0,2\pi R) \,, \ t \in (-a,a) 
  \Big\}
  \,,
\end{multline}
where $a \in (0,R)$ is the half-width of the strip,
see Figure~\ref{figMobius}.
Notice that the surface~$\Omega$ is not orientable due to the division
of the angle~$s/R$ by the factor two, 
while the strip is still well ``glued together'' at the endpoints
corresponding to $s=0$ and $s=2\pi R$.
We consider the self-adjoint operator  
\begin{equation*}
  -\Delta_D^\Omega
  \qquad \mbox{in} \qquad
  \sii(\Omega)
  \,,
\end{equation*}
which acts as Laplace--Beltrami operator in~$\Omega$
and satisfies Dirichlet boundary conditions on~$\partial\Omega$. 
Depending on whether we consider the wave, heat or Schr\"odinger 
equation on~$\Omega$, the eigenvalues and eigenfunctions of $-\Delta_D^\Omega$ 
have various physical meaning.
Here we mostly use the quantum-mechanical language, 
where $-\Delta_D^\Omega$ is the Hamiltonian of an electron
constrained to~$\Omega$ by hard-wall boundaries
and the eigenvalues and eigenfunctions correspond to
bound-state energies and wave-functions, respectively.
(In order to simultaneously consider the other physical models, 
we disregard the possibility of adding a geometric potential
due to the embedding of~$\Omega$ in~$\Real^3$,
see~\cite{KRT} and references therein.)
When interested in analytic properties of the M\"obius strip~$\Omega$,
we call $-\Delta_D^\Omega$ the \emph{true} (or \emph{full}) {model}.

\begin{figure}
  {\centering
    \includegraphics{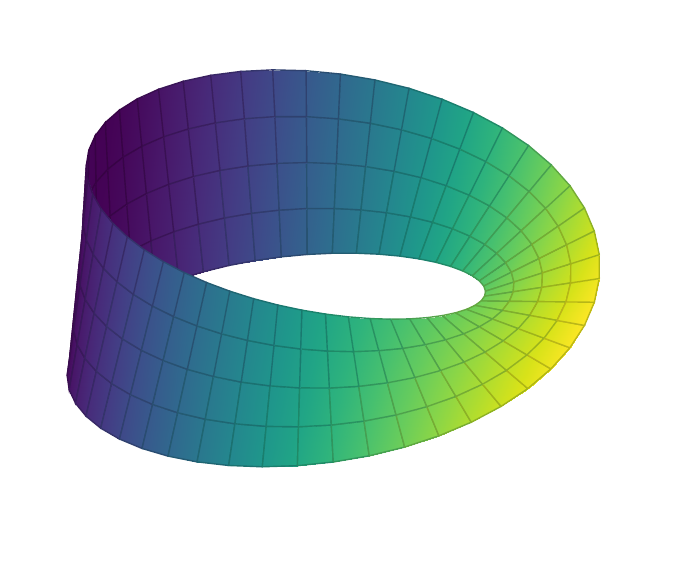}
  \par}
  \caption{M\"obius strip $\Omega$ built along a circle of radius $R$ and with width $2a$, see~\eqref{true}.}
  \label{figMobius}
\end{figure}

\paragraph{The fake model.}
The weak point of the true model $-\Delta_D^\Omega$ 
is that its spectrum cannot be computed explicitly
(for numerical calculations, 
see, \eg, \cite{Macdonald-Brandman-Ruuth_2011, LiRamMohan}, and below).
On the other hand, the spectrum of the \emph{flat} model
$-\Delta_D^{\Omega_0}$ represented by the Laplacian in%
\begin{equation*} 
  \Omega_0 := (0,2\pi R) \times (-a,a) 
  \,,
\end{equation*}
subject to Dirichlet boundary conditions 
\begin{equation}\label{bc.Dirichlet}
  \psi(s,\pm a) = 0 
\end{equation}
for almost every $s \in (0,2\pi R)$
and twisted periodic boundary conditions
\begin{align}
  \psi(0,t) &= \psi(2 \pi R,-t) \,, 
  \label{bc1}
  \\ 
  \partial_1\psi(0,t) &= \partial_1\psi(2 \pi R,-t) \,,
  \label{bc2}
\end{align}
for almost every $t \in (-a,a)$, is explicitly computable. 
Indeed, by considering the spectral problem 
(which can be solved by separation of variables)
for the Laplacian in the extended strip $(-2\pi R,2\pi R) \times (-a,a)$,
subject to Dirichlet boundary conditions on $(0,2\pi R)\times\{\pm a\}$
and standard periodic boundary conditions 
on $\{\pm 2\pi R\} \times (-a,a)$,
by symmetry arguments one easily arrives at
\begin{equation}\label{spec.fake}
  \sigma(-\Delta_D^{\Omega_0}) = 
  \left\{
  \left(\frac{m}{2 R}\right)^2 + \left(\frac{n\pi}{2a}\right)^2
  \right\}_{\subalign{m &\in \Int \\ n &\in \Nat^* \\ m+n &\text{ odd}}}
  ,
\end{equation}
where $\Nat^* := \Nat \setminus \{0\}$ 
(the set $\Nat$ of all natural numbers contain zero in our convention).
At the same time, the eigenfunctions corresponding 
to the eigenvalues in~\eqref{spec.fake} are explicit 
combinations of sines and cosines. 
While the flat model $-\Delta_D^{\Omega_0}$ 
keeps the topology of the M\"obius strip,
it completely disregards the curved nature of~$\Omega$,
and for this reason we also call it the \emph{fake} model.  
Of course, there is no reason to expect that the spectrum~\eqref{spec.fake}
is in any sense related to (\ie, approximating)
the spectrum of $-\Delta_D^{\Omega}$.

\paragraph{The not-so-fake model.}
The main objective of this paper is to introduce a new model
of the M\"obius strip, whose spectrum is explicitly computable 
and simultaneously approximates the spectrum of the true model $-\Delta_D^{\Omega}$
in the limit of \emph{thin} strips, \ie\ $a \to 0$. 
This \emph{not-so-fake} (or \emph{effective})
model is given by the perturbed operator
\begin{equation}\label{eff}
  H_\mathrm{eff} := -\Delta_D^{\Omega_0} + V_\mathrm{eff}
  \qquad \mbox{with} \qquad
  V_\mathrm{eff}(s,t) := -\frac{1}{8 R^2} \, \cos\left(\frac{s}{R}\right) 
  \,.
\end{equation}
Hence, the addition of the geometric potential~$V_\mathrm{eff}$
to the fake model is necessary to get 
a more realistic approximation of the dynamics.
In accordance with general thin strips \cite[Sec.~3]{K1},
we have
$V_\mathrm{eff}(s,t)=-\frac{1}{4} \kappa_g(s)^2 - \frac{1}{2} K(s,0)$,
where~$K$ is the Gauss curvature of~$\Omega$ 
and~$\kappa_g$ is the geodesic curvature of the underlying circle
as a curve on the surface~$\Omega$, see Remark~\ref{Rem.Fermi}.   
The aforementioned approximation in thin strips 
will be justified by showing that the norm of the difference of the resolvents
of $-\Delta_D^{\Omega}$ and~$H_\mathrm{eff}$ vanishes as $a \to 0$.
Two remarks are in order.
First, since the operators act in different Hilbert spaces,
a natural identification is necessary.
Second, since the eigenvalues  
tend to infinity as $a \to 0$, see~\eqref{spec.eff} below,
it does not really make sense to compare the resolvents
$(-\Delta_D^{\Omega}-z)^{-1}$ and $(H_\mathrm{eff}-z)^{-1}$ 
(at least with an $a$-independent~$z$);
instead we consider renormalised operators 
by subtracting the lowest transverse energy 
\begin{equation*}
 E_1 := \left(\frac{\pi}{2a}\right)^2 \,.
\end{equation*}
Namely, we prove (see Theorem~\ref{Thm.nrs})
\begin{equation}\label{nrs}
  \big\| U(-\Delta_D^{\Omega}-E_1-z)^{-1}U^{-1} 
  - (H_\mathrm{eff}-E_1-z)^{-1} \big\| = O(a) 
  \qquad \mbox{as} \qquad 
  a \to 0
  \,,
\end{equation}
where $U : \sii(\Omega) \to \sii(\Omega_0)$
is a suitable unitary transform
and~$z$ is any $a$-independent number lying simultaneously 
in the resolvent sets of the shifted operators
$-\Delta_D^{\Omega}-E_1$ and $H_\mathrm{eff}-E_1$.
In other words,  
we establish the norm-resolvent convergence
in a generalised sense.
An advantage of this approximation is that
(by the same extension trick as for the fake model above)
the spectrum of~$H_\mathrm{eff}$ can be computed explicitly:
\begin{equation}\label{spec.eff}
  \sigma(H_\mathrm{eff}) = 
  \left\{ 
  \left(\frac{1}{2 R}\right)^2 a_m\!\left(-\frac{1}{4}\right)
  + \left(\frac{n\pi}{2a}\right)^2 
  \right\}_{\subalign{m &\in \Nat \\ n&\in\Nat^* \\ m+n & \text{ odd}}}
  \ \cup \
  \left\{  
  \left(\frac{1}{2 R}\right)^2 b_m\!\left(-\frac{1}{4}\right)
  + \left(\frac{n\pi}{2a}\right)^2
  \right\}_{\subalign{m&\in\Nat^* \\ n&\in\Nat^* \\ m+n & \text{ odd}}}
  ,
\end{equation}
where~$a_m$ and~$b_m$ are Mathieu characteristic values
(see Remark~\ref{remark.Mathieu} below and \cite[Sec.~20]{Abramowitz-Stegun}).
The eigenfunctions corresponding 
to the eigenvalues in~\eqref{spec.eff} are this time explicit 
combinations of Mathieu integral order functions. 

\paragraph{Organisation of the paper.}
The fake and not-so-fake models are studied 
in Sections~\ref{Sec.fake} and~\ref{Sec.eff}, respectively,
where we particularly establish 
the spectral results~\eqref{spec.fake} and~\eqref{spec.eff}
and determine the corresponding eigenfunctions.
The norm-resolvent convergence~\eqref{nrs} 
is established in Section~\ref{Sec.nrs}.
Our analytical results are illustrated by numerical computations
in Section~\ref{Sec.numerics}.
In particular, we show the suboptimality of the rate $O(a)$ 
in the approximation of the spectrum of~$-\Delta_D^\Omega$
by the spectrum of the not-so-fake model~$H_\mathrm{eff}$.

%-----------------------%
\section{The fake model}\label{Sec.fake}
%-----------------------%
%
The fake model is introduced as the operator 
$-\Delta_D^{\Omega_0}$ in $\sii(\Omega_0)$ defined by 
\begin{equation*}
\begin{aligned}
  -\Delta_D^{\Omega_0} \psi 
  &:= -\Delta\psi \,,
  \\
  \Dom(-\Delta_D^{\Omega_0})
  &:= \left\{ \psi \in W^{2,2}(\Omega_0) : \
  \mbox{$\psi$ satisfies \eqref{bc.Dirichlet} and \eqref{bc1}--\eqref{bc2}}
  \right\} \,,
\end{aligned}
\end{equation*}
where the boundary conditions are understood
in the sense of Sobolev-space traces (see, \eg, \cite{Adams2}).

\begin{Proposition}
$-\Delta_D^{\Omega_0}$ is a positive self-adjoint operator
with compact resolvent.
\end{Proposition}
\begin{proof}
The operator is clearly densely defined.
Let $\|\cdot\|$ and $(\cdot,\cdot)$ denote the norm 
and inner product of $\sii(\Omega_0)$.
For every $\psi \in \Dom(-\Delta_D^{\Omega_0})$, we have
\begin{equation*}
\begin{aligned}
  (\psi,-\Delta_D^{\Omega_0}\psi)
  &= \|\nabla\psi\|^2 
  - \int_{-a}^{a} 
  \Big[\overline{\psi}(s,t)\,\partial_1\psi(s,t)\Big]_{s=0}^{s=2 \pi R} 
  \, \der t
  - \int_{0}^{2 \pi R} 
  \Big[\overline{\psi}(s,t)\,\partial_2\psi(s,t)\Big]_{t=-a}^{t=a} 
  \, \der s
  \\
  &= \|\nabla\psi\|^2 
  \,,
\end{aligned}
\end{equation*}
where the second boundary integral on the first line 
vanishes because of~\eqref{bc.Dirichlet} and, 
using \eqref{bc1}--\eqref{bc2} together with an integral substitution,
\begin{equation*}
\begin{aligned}
  \int_{-a}^{a} 
  \Big[\overline{\psi}(s,t)\,\partial_1\psi(s,t)\Big]_{s=0}^{s=2 \pi R} 
  \, \der t
  &= \int_{-a}^{a} 
  \overline{\psi}(2\pi R,t)\,\partial_1\psi(2\pi R,t) 
  \, \der t
  - 
  \int_{-a}^{a} 
  \overline{\psi}(0,t)\,\partial_1\psi(0,t) 
  \, \der t
  \\
  &= \int_{-a}^{a} 
  \overline{\psi}(0,-t)\,\partial_1\psi(0,-t) 
  \, \der t
  - 
  \int_{-a}^{a} 
  \overline{\psi}(0,t)\,\partial_1\psi(0,t) 
  \, \der t
  \\
  &= \int_{-a}^{a} 
  \overline{\psi}(0,u)\,\partial_1\psi(0,u) 
  \, \der u
  - 
  \int_{-a}^{a} 
  \overline{\psi}(0,t)\,\partial_1\psi(0,t) 
  \, \der t
  = 0
  \,.
\end{aligned} 
\end{equation*}
Hence, $-\Delta_D^{\Omega_0}$ is obviously symmetric.
At the same time, using the Poincar\'e-type inequality
\begin{equation}\label{Poincare}
  \forall f \in W_0^{1,2}((-a,a))
  \,, \qquad
  \int_{-a}^a |f'(t)|^2 \, \der t
  \geq  E_1
  \int_{-a}^a |f(t)|^2 \, \der t
  \,,
\end{equation}
we have 
\begin{equation}\label{Poincare.consequence}
  (\psi,-\Delta_D^{\Omega_0}\psi)
  = \|\nabla\psi\|^2 
  \geq \|\partial_2\psi\|^2
  \geq E_1 \|\psi\|^2
  \,,
\end{equation}
so $-\Delta_D^{\Omega_0}$ is a positive operator
whose spectrum does not start below~$E_1$.
To see that $-\Delta_D^{\Omega_0}$ is self-adjoint is more subtle.
A possibility how to verify it is to consider the closed quadratic form
\begin{equation}\label{form.fake}
\begin{aligned}
  Q_D^{\Omega_0}[\psi] 
  &:= \|\nabla\psi\|^2 \,,
  \\
  \Dom(Q_D^{\Omega_0})
  &:= \left\{ \psi \in W^{1,2}(\Omega_0) : \
  \mbox{$\psi$ satisfies \eqref{bc.Dirichlet} and \eqref{bc1}}
  \right\} \,.
\end{aligned}
\end{equation}
Let~$H_0$ denote the self-adjoint operator associated with $Q_D^{\Omega_0}$
via the first representation theorem \cite[Thm.~VI.2.1]{Kato}.
We claim that $H_0 = -\Delta_D^{\Omega_0}$.
Indeed, by an integration by parts,
it is easy to see that $-\Delta_D^{\Omega_0} \subset H_0$.
The opposite inclusion $-\Delta_D^{\Omega_0} \supset H_0$
can be checked with help of standard elliptic regularity theory
(\cf~\cite[Sec.~3]{BK} for an analogous problem).
Finally, the resolvent of $-\Delta_D^{\Omega_0}$ is compact
as a consequence of the compactness of the Sobolev embedding
$W^{2,2}(\Omega_0) \hookrightarrow \sii(\Omega_0)$.
\end{proof}

In order to determine the spectrum of $-\Delta_D^{\Omega_0}$,
it is convenient to consider the extended strip 
\begin{equation*}
  \tilde{\Omega}_0 := (-2\pi R,2\pi R) \times (-a,a)
\end{equation*}
and the associated Laplacian~$T_0$ in $\sii(\tilde{\Omega}_0)$
with combined Dirichlet and standard periodic boundary conditions,
\begin{equation*}
\begin{aligned}
  T_0\psi 
  &:= -\Delta\psi \,,
  \\
  \Dom(T_0)
  &:= \Big\{ \psi \in W^{2,2}(\tilde{\Omega}_0) : \
  \mbox{$\psi$ satisfies} \quad
  \psi(s,\pm a) = 0 \,, \quad
  \forall s \in (-2\pi R,2\pi R) \,,
  \\
  & \qquad\qquad
  \mbox{and} \quad
  \psi(-2 \pi R,t) = \psi(2 \pi R,t) \,, \quad
  \partial_1\psi(-2 \pi R,t) = \partial_1\psi(2 \pi R,t) \,, \quad
  \forall t \in (-a,a) 
  \Big\} \,.
\end{aligned}
\end{equation*}
This operator is well known to be self-adjoint 
and its spectrum can be easily found by separation of variables:
\begin{equation*}%\label{spec.per}
  \sigma(T_0) = \left\{
  \left(\frac{m}{2R}\right)^2 + \left(\frac{n\pi}{2a}\right)^2
  \right\}_{\subalign{m &\in \Int \\ n &\in \Nat^*}}
  \,.
\end{equation*}
The corresponding normalised eigenfunctions of~$T_0$ read
\begin{equation*}
  \phi_{m,n}(s,t) = \varphi_m(s) \, \chi_n(t) 
\end{equation*}
with
\begin{equation}\label{chi}
  \varphi_m(s) := \sqrt{\frac{1}{4 \pi R}} \, e^{i\frac{m}{2R}s}
  \,, \qquad
  \chi_n(t) :=
\begin{cases}
  \sqrt{\frac{1}{a}} \,
  \cos(\frac{n\pi}{2a}t)
  & \mbox{if $n$ is odd}
  \,, \\
  \sqrt{\frac{1}{a}} \,
  \sin(\frac{n\pi}{2a}t)
  & \mbox{if $n$ is even}
  \,.
\end{cases}
\end{equation}
As eigenfunctions of a self-adjoint operator,
it is also well known that 
$\{\phi_{m,n}\}_{m \in \Int,\, n \in \Nat^*}$
is a complete orthonormal set in $\sii((-2\pi R,2 \pi R) \times (-a,a))$.

By symmetry properties of~$\varphi_m$ and~$\chi_n$,
we have
\begin{equation*}
\begin{aligned}
  \phi_{m,n}(2\pi R,-t) &= (-1)^{m+n+1} \, \phi_{m,n}(0,t)
  \,, \\
  \partial_1\phi_{m,n}(2 \pi R,-t) &= (-1)^{m+n+1} \, \partial_1\phi_{m,n}(0,t)
  \,.
\end{aligned}
\end{equation*}
Therefore we see that~$\phi_{m,n}$ satisfies the boundary conditions 
\eqref{bc1}--\eqref{bc2} if, and only if, $m+n$ is odd.
Consequently,
\begin{equation}\label{inclusion}
  \sigma(-\Delta_D^{\Omega_0}) \supset 
  \left\{
  \left(\frac{m}{2 R}\right)^2 + \left(\frac{n\pi}{2a}\right)^2
  \right\}_{\subalign{m &\in \Int \\ n &\in \Nat^* \\ m+n &\text{ odd}}}
\end{equation}
and the corresponding normalised eigenfunctions of~$-\Delta_D^{\Omega_0}$ 
are given by the restrictions
\begin{equation}\label{eq.psimn}
  \psi_{m,n} := \sqrt{2} \, \phi_{m,n} \upharpoonright \Omega_0
  \,, \qquad
  m \in \Int, \ n \in \Nat^*, \ m+n \text{ is odd}
  \,.
\end{equation}
To show that the right-hand side of~\eqref{inclusion} 
determines \emph{all} the eigenvalues of~$-\Delta_D^{\Omega_0}$, 
we need the following result. 

\begin{Proposition}\label{prop.flat}
$\{\psi_{m,n}\}_{m \in \Int,\, n \in \Nat^*,\, m+n \text{ \emph{odd}}}$
is a complete orthonormal set in $\sii(\Omega_0)$.
\end{Proposition}
\begin{proof}
The property that
$\{\phi_{m,n}\}_{m \in \Int,\, n \in \Nat^*}$
is a complete orthonormal set in $\sii(\tilde{\Omega}_0)$
is equivalent to the validity of the Parseval equality
\begin{equation}\label{Parseval.per}
  \|f\|^2 = \sum_{m \in \Int, \, n \in \Nat^*} |(\phi_{m,n},f)|^2
\end{equation}
for every $f \in \sii(\tilde{\Omega}_0)$,
where~$\|\cdot\|$ and $(\cdot,\cdot)$ 
denote the norm and inner product of $\sii(\tilde{\Omega}_0)$, respectively.
Given an arbitrary $g \in \sii(\Omega_0)$,
we define the extension
\begin{equation}\label{extension}
  f(s,t) :=
  \begin{cases}
    g(s,t) & \mbox{if} \quad s > 0 \,,
    \\
    g(s+2\pi R,-t) & \mbox{if} \quad s < 0 \,.
  \end{cases}
\end{equation}
By an obvious integral substitution, 
it is straightforward to check the identity
\begin{equation}\label{id1}
  \|f\|^2 = 2 \, \|g\|^2 \,,
\end{equation}
where on the right-hand side
we use the same notation~$\|\cdot\|$ for the norm of $\sii(\Omega_0)$.
At the same time,
using in addition to the substitution
the symmetry properties of~$\varphi_m$ and~$\chi_n$,
we have
\begin{equation}\label{id2}
  (\phi_{m,n},f) = \frac{1}{\sqrt{2}} \ [1+(-1)^{m+n+1}] \, (\psi_{m,n},g)
  \,,
\end{equation}
where on the right-hand side
we use the same notation~$(\cdot,\cdot)$ for the inner product of $\sii(\Omega_0)$.
Putting~\eqref{id1} and~\eqref{id2} into~\eqref{Parseval.per},
we get the Parseval equality
\begin{equation*}\label{Parseval.Mobius}
  \|g\|^2 = \sum_{\subalign{&m \in \Int, \, n \in \Nat^* \\ &\,m+n \text{ odd}}} 
  |(\psi_{m,n},g)|^2 
  \,,
\end{equation*}
which is equivalent to the desired completeness result.
\end{proof}

As a consequence of this proposition and~\eqref{inclusion},
we conclude with the desired result~\eqref{spec.fake}.
The elements of the set~\eqref{spec.fake} 
will be denoted by~$\lambda_{m,n}$.

\begin{Remark}
The lowest eigenvalue 
\begin{equation*}
  \lambda_{0,1} = E_1
  %= \left(\frac{\pi}{2a}\right)^2
\end{equation*}
is simple and the corresponding eigenfunction~$\psi_{0,1}$ is positive. 
In particular, the inequality~\eqref{Poincare.consequence} is optimal.
The eigenvalues~$\lambda_{m,n}$ with $m \not=0$ are always degenerate.   
In particular, the second eigenvalue
\begin{equation*}
  \min\{\lambda_{1,2}=\lambda_{-1,2},\lambda_{2,1}=\lambda_{-2,1}\}
\end{equation*}
is always degenerate. 
Furthermore, if $\pi R=a$ then the second eigenvalue has multiplicity four.
\end{Remark}
%

%------------------------------%
\section{The not-so-fake model}\label{Sec.eff}
%------------------------------%
%
Since~$V_\mathrm{eff}$ is real-valued and bounded,
the operator sum $H_\mathrm{eff} := -\Delta_D^{\Omega_0} + V_\mathrm{eff}$
in~\eqref{eff} defines a self-adjoint operator. 
As in the fake model, together with~$H_\mathrm{eff}$ in $\sii(\Omega_0)$,
we also consider $T_\mathrm{eff} := T_0 + V_\mathrm{eff}$ 
in the extended Hilbert space $\sii(\tilde{\Omega}_0)$.
The eigenvalues and eigenfunctions of~$T_\mathrm{eff}$ 
can be found by separation of variables.
In the second variable we get the same result as in the previous section:
the normalised eigenfunctions of the Laplacian in $\sii((-a,a))$,
subject to Dirichlet boundary conditions, 
are numbered by $n \in \Nat^*$ and given by~$\chi_n$ as in~\eqref{chi}
and the corresponding eigenvalues are $E_1 n^2$.
The case of the first variable is a little bit more involved.
After the separation of variables,
we arrive at the differential equation
\begin{equation}\label{eq.potential.s}
  -\varphi''(s) 
  - \frac{1}{8 R^2} \cos\left(\frac{s}{R}\right) \varphi(s) = \nu \varphi(s)\,.
\end{equation}
It turns out that this is the Mathieu differential equation.
Before we proceed any further let us first review basic properties of Mathieu functions.

\begin{Remark}[Mathieu functions]\label{remark.Mathieu}
We use the following notation (see~\cite[\S28.2]{NIST}).
Fix $q, \mu \in \Real$ and consider the ordinary differential equation
\begin{equation}\label{eq.mathieu}
  y''(\eta) + \big(\mu - 2q\cos(2\eta)\big) y(\eta) = 0\,.
\end{equation}
This equation has a $2\pi$-periodic solution if, and only if,
$\mu = a_m(q)$ or $\mu = b_m(q)$, 
where $a_m(q)$ with $m\in\Nat$
and $b_m(q)$ with $m\in\Nat^*$ 
are the so-called Mathieu characteristic values.
These characteristic values satisfy
\begin{align}
  q &> 0: \quad a_0 < b_1 < a_1 < b_2 < a_2 < \cdots\,, \nonumber \\
  q &< 0: \quad a_0 < a_1 < b_1 < b_2 < a_2 < \cdots\,, \label{eq.mathieu.ab} \\
  q &= 0: \quad a_m(0) = b_m(0) = m^2\,. \nonumber
\end{align}
The Mathieu integral order functions $\ce_m(\eta, q)$ with $m\in\Nat$
and $\se_m(\eta, q)$ with $m\in\Nat^*$ are defined in the following way: 
$\ce_m(\eta, q)$ is the even solution of~\eqref{eq.mathieu} with $\mu = a_m(q)$ 
and $\se_m(\eta, q)$ is the odd solution of~\eqref{eq.mathieu} with $\mu = b_m(q)$.
Both $\ce_m(\cdot, q)$ and $\se_m(\cdot, q)$ are $2\pi$-periodic.
Moreover, $\ce_{2m}(\cdot, q)$ and $\se_{2m+2}(\cdot, q)$ 
are $\pi$-periodic and $\ce_{2m+1}(\cdot, q)$ and $\se_{2m+1}(\cdot, q)$ 
are antiperiodic with antiperiod~$\pi$.
For any $q \in \Real$, the integral order Mathieu functions 
$\ce_m(\eta, q)$ and $\se_m(\eta, q)$ taken together 
form an orthogonal basis in $\sii((-\pi, \pi))$ 
(see~\cite[\S 20.5]{Abramowitz-Stegun}).
We assume both $\ce_m(\eta,q)$ and $\se_m(\eta,q)$ 
are normalised to $\sqrt{\pi}$ in $\sii((-\pi, \pi))$, 
\ie\ the equalities
\begin{equation}\label{final1}
  \int_{-\pi}^\pi \big|\ce_m(\eta,q)\big|^2\,\mathrm{d}\eta 
  = \int_{-\pi}^\pi \big|\se_m(\eta,q)\big|^2\,\mathrm{d}\eta 
  = \pi
\end{equation}
hold for all possible values of~$m$. 
This convention is in agreement with~\cite{NIST} 
and it is respected by Wolfram Mathematica, too.
The (anti)periodicity then implies
\begin{equation}\label{final2}
  \int_0^\pi \big|\ce_m(\eta,q)\big|^2\,\mathrm{d}\eta 
  = \int_0^\pi \big|\se_m(\eta,q)\big|^2\,\mathrm{d}\eta 
  = \frac{\pi}{2}\,.
\end{equation}
\end{Remark}

Let us now return to the equation~\eqref{eq.potential.s}.
Employing a simple change of the independent variable, $\eta = s/(2R)$, 
we immediately get the Mathieu equation~\eqref{eq.mathieu} 
with $q = -1/4$ and $\mu = 4 R^2 \nu$.
Thus the equation~\eqref{eq.potential.s} has the following 
$(4 \pi R)$-periodic 
and normalised solutions if and only if $\nu$ satisfies one of the indicated conditions
\begin{align}
  \label{eq.varphi1}
  \varphi^{(1)}_m(s) 
  &:= \frac{1}{\sqrt{\pi}} \se_m\left(\frac{s}{2R}, -\frac{1}{4}\right) &&
  \text{if} \quad 
  b_m\left(-\frac{1}{4}\right) = 4 R^2 \nu \quad
  \text{for some} \quad m \in \Nat^*, \\
  \label{eq.varphi2}
  \varphi^{(2)}_m(s) 
  &:= \frac{1}{\sqrt{\pi}} \ce_m\left(\frac{s}{2R}, -\frac{1}{4}\right) &&
  \text{if} \quad 
  a_m\left(-\frac{1}{4}\right) = 4 R^2 \nu  \quad
  \text{for some} \quad m \in \Nat\,.
\end{align}
The eigenvalues of~$T_\mathrm{eff}$ therefore read
\[
  \sigma(T_\mathrm{eff}) = 
  \left\{ 
  \left(\frac{1}{2R}\right)^2 a_m\left(-\frac{1}{4}\right)
  + \left(\frac{n\pi}{2a}\right)^2   
  \right\}_{\subalign{m&\in\Nat \\ n&\in\Nat^*}}
  \ \cup \
  \left\{ 
  \left(\frac{1}{2R}\right)^2 b_m\left(-\frac{1}{4}\right)
  + \left(\frac{n\pi}{2a}\right)^2 
  \right\}_{\subalign{m&\in\Nat^* \\ n&\in\Nat^*}}\,.
\]
The corresponding normalised eigenfunctions are given by 
\begin{align*}
  \phi^{(1)}_{m,n}(s,t) 
  &:= \varphi^{(1)}_m(s) \, \chi_n(t) \,, \qquad m\in\Nat^*, \ n\in\Nat^*, 
  \\
  \phi^{(2)}_{m,n}(s,t) 
  &:= \varphi^{(2)}_m(s) \, \chi_n(t) \,, \qquad m\in\Nat\,, \ n\in\Nat^*,
\end{align*}
and they form a complete orthonormal set of $\sii(\tilde{\Omega}_0)$.

Let us now find the eigenfunctions and eigenvalues of 
the not-so-fake M\"obius-strip operator~$H_\mathrm{eff}$.
Note that for any $j=1,2$ the functions $\varphi^{(j)}_m$ 
are antiperiodic (respectively, periodic)
with antiperiod $2 \pi R$ (respectively, period $2 \pi R$)
whenever $m$ is odd (respectively, even). 
Using this observation we establish 
the following symmetry properties of the eigenfunctions of $T_\mathrm{eff}$: 
\begin{equation}\label{eq.phi.symm}
  \phi^{(j)}_{m,n}(s+2\pi R,-t) 
  = \varphi^{(j)}_m(s+2\pi R) \chi_n(-t) 
  = (-1)^m \varphi^{(j)}_m(s) (-1)^{n+1} \chi_n(t) 
  = (-1)^{m+n+1} \phi^{(j)}_{m,n}(s,t)
  \,,
\end{equation}
for any $j=1,2$ and all permissible~$m$ and~$n$.
In particular, setting $s=0$ in the last equation we have
\begin{equation*}
  \phi^{(j)}_{m,n}(2\pi R,-t) = (-1)^{m+n+1} \phi^{(j)}_{m,n}(0,t)
\end{equation*}
and so~$\phi^{(j)}_{r,n}$ with $j=1,2$ satisfies 
the boundary conditions~\eqref{bc1}--\eqref{bc2} 
if, and only if, $m+n$ is odd.
Consequently,
\begin{equation}\label{inclusion.potential}
  \sigma(H_\mathrm{eff}) \supset
  \left\{ 
  \left(\frac{1}{2R}\right)^2 a_m\left(-\frac{1}{4}\right)
  + \left(\frac{n\pi}{2a}\right)^2   
  \right\}_{\subalign{m&\in\Nat \\ n&\in\Nat^* \\ m+n & \text{ odd}}}
  \ \cup \
  \left\{ 
  \left(\frac{1}{2R}\right)^2 b_m\left(-\frac{1}{4}\right)
  + \left(\frac{n\pi}{2a}\right)^2 
  \right\}_{\subalign{m&\in\Nat^* \\ n&\in\Nat^* \\ m+n & \text{ odd}}}
  \,.
\end{equation}
The corresponding normalised eigenfunctions of~$H_\mathrm{eff}$ 
are given by the restrictions
\begin{equation*}
  \psi^{(j)}_{m,n} := \sqrt{2} \, \phi^{(j)}_{m,n} \upharpoonright \Omega_0
  \,,
\end{equation*}
where $(m,n) \in \Nat^* \times \Nat^* =: \Nat_1$ if $j = 1$ 
and $(m,n) \in \Nat \times \Nat^* =: \Nat_2$ if $j=2$.
That the normalisation factor $\sqrt{2}$ is correct follows from 
the final equations~\eqref{final1} and~\eqref{final2} in Remark~\ref{remark.Mathieu} 
and equations \eqref{eq.varphi1}--\eqref{eq.varphi2}.
 
To show that the right-hand side of~\eqref{inclusion.potential} 
determines \emph{all} the eigenvalues of~$H_{\mathrm{eff}}$, we need the following result analogous to Proposition~\ref{prop.flat}.
\begin{Proposition}
$\big\{\psi^{(j)}_{m,n}\big\}_{j=1,2, \, (m,n)\in\Nat_j, \, m+n \text{\emph{ odd}}}$
is a complete orthonormal set in $\sii(\Omega_0)$.
\end{Proposition}
\begin{proof}
The property that the set $\big\{\phi^{(j)}_{m,n}\big\}_{j=1,2, \, (m,n)\in\Nat_j}$
is a complete orthonormal set in $\sii(\tilde{\Omega}_0)$
is equivalent to the validity of the Parseval equality
\begin{equation}\label{Parseval.per2}
  \|f\|^2 = \sum_{\subalign{&(m,n) \in \Nat_j \\ &\ j = 1,2 \ }} 
  \big|\big(\phi^{(j)}_{m,n},f\big)\big|^2
\end{equation}
for every $f \in \sii(\tilde{\Omega}_0)$.
Given an arbitrary $g \in \sii(\Omega_0)$,
we define the extension $f \in \sii(\tilde{\Omega}_0)$ 
as in~\eqref{extension}.
By an obvious integral substitution, 
it is straightforward to check the identity~\eqref{id1}.
At the same time,
using in addition to the substitution
the symmetry property~\eqref{eq.phi.symm},
we have
\begin{align}
  \big(\phi^{(j)}_{m,n},\, f\big) 
  &= \int_{(-2 \pi R,0) \times (-a,a)} \phi^{(j)}_{m,n}(s,t) \, g(s+2 \pi R,-t) 
  \,\mathrm{d}s \, \mathrm{d}t 
  + \int_{(0,2 \pi R) \times (-a,a)} \phi^{(j)}_{m,n}(s,t) \, g(s,t) 
  \,\mathrm{d}s \, \mathrm{d}t \nonumber
  \\
  &= \int_{(0,2 \pi R) \times (-a,a)} \phi^{(j)}_{m,n}(s-2\pi R,-t) \, g(s,t) 
  \,\mathrm{d}s\,\mathrm{d}t 
  + \frac{1}{\sqrt{2}} \int_{(0,2\pi R) \times (-a,a)} \psi^{(j)}_{r,n}(s,t) \, g(s,t) 
  \,\mathrm{d}s \, \mathrm{d}t \nonumber
  \\
  &= \frac{(-1)^{m+n+1}}{\sqrt{2}}
  \int_{(0,2\pi R) \times (-a,a)} 
  \psi^{(j)}_{m,n}(s,t) \, g(s,t) 
  \,\mathrm{d}s \, \mathrm{d}t 
  + \frac{1}{\sqrt{2}} 
  \int_{(0,2\pi R) \times (-a,a)} 
  \psi^{(j)}_{m,n}(s,t) \, g(s,t) 
  \,\mathrm{d}s \, \mathrm{d}t 
  \nonumber\\
  &= \frac{1}{\sqrt{2}} \big[ (-1)^{m+n+1} + 1 \big] 
  \big(\psi^{(j)}_{m,n}, \, g \big)\,. 
  \label{id4}
\end{align}
Putting~\eqref{id1} and~\eqref{id4} into~\eqref{Parseval.per2},
we get the Parseval inequality
\begin{equation*}\label{Parseval.Mobius2}
  \|g\|^2 = \sum_{\subalign{&(m,n) \in \Nat_j \\ &\ j=1,2 \ \\ & m+n \text{ odd}}}
  \big|\big(\psi^{(j)}_{m,n},\,g\big)\big|^2 
  \,,
\end{equation*}
which is equivalent to the desired completeness result.
\end{proof}

As a consequence of this proposition and~\eqref{inclusion.potential},
we conclude with the desired result~\eqref{spec.eff}.
Note that the particular Mathieu characteristic values $a_m(-1/4)$ and $b_m(-1/4)$ appearing in~\eqref{spec.eff} do not depend on $a$ neither $R$.
However, their value gets very close to each other with increasing $m$.
First few of these values are presented in Table~\ref{tab.mathieu.characteristics}.
Consequently, the not-so-fake model exhibits pairs of eigenvalues located very close each other (also recall~\eqref{eq.mathieu.ab}).

\begin{table}
  {\centering
  \begin{tabular}{lr@{}lr@{}l}
  $m$ &   & $a_m$                                 &  &  $b_m$ \\ \hline
   0  &  -0&.03103939547561732443850972818046737540 & & \\
   1  &  0&.74242882598662974339949054767095543815 &   1&.24194112824291514482231057477841662622 \\
   2  &  4&.02582908464560324171350493521402514557 &   3&.99479307863211894594328093443536761399 \\
   3  &  9&.00366486704623913463365662695182921571 &   9&.00415255154693478030510107620470513307 \\
   4  &  16&.00208529046719562998287970766353836899 &  16&.00208190103817298727073812993351765300 \\
   5  &  25&.00130213222684081366209108945453834337 &  25&.00130214546980228095721811268235655121 \\
   6  &  36&.00089287379843422726407677439950789279 &  36&.00089287376532391463296827349981967276 \\
   7  &  49&.00065104784806396399969278784780613747 &  49&.00065104784812144953869393158610105146 \\
   8  &  64&.00049603440671169350384368118283820869 &  64&.00049603440671162017886328541877470187 \\
   9  &  81&.00039062627570760760462351056102476286 &  81&.00039062627570760767623083270127588410 \\
  10  &  100&.00031565723007867410511381290959992431 & 100&.00031565723007867410505855991940003139
  \end{tabular}
  \par}
  \caption{Mathieu characteristic values $a_m(-1/4)$ and $b_m(-1/4)$
  occurring in the spectrum~\eqref{spec.eff} of the not-so-fake model computed by Wolfram Mathematica.}
  \label{tab.mathieu.characteristics}
\end{table}

%-----------------------------------------------%
\section{From the true to the not-so-fake model}\label{Sec.nrs}
%-----------------------------------------------%
%
By definition~\eqref{true}, 
$\Omega = \mathscr{L}([0,2\pi R)\times(-a,a))$,
where $\mathscr{L}:\Real^2\to\Real^3$ is given by 
\begin{equation*}
  \mathscr{L}(s,t) :=
  \left( 
  \left[R-t \cos\left(\frac{s}{2R}\right)\right] 
  \cos\left(\frac{s}{R}\right), 
  \left[R-t \cos\left(\frac{s}{2R}\right)\right] 
  \sin\left(\frac{s}{R}\right), 
  - t \sin\left(\frac{s}{2 R}\right)
  \right)
  \,.
\end{equation*}
Except for the segment $\mathscr{L}(\{0\}\times(-a,a))$, 
which has the Lebesgue measure equal to zero,
it is thus possible to identify~$\Omega$ 
with the Riemannian manifold $(\Omega_0,G)$,
where $G := \nabla\mathscr{L} \cdot (\nabla\mathscr{L})^T$
is the metric induced by~$\mathscr{L}$. 
It is straightforward to check that~$G$ has the diagonal form
\begin{equation*}\label{metric}
  G = 
  \begin{pmatrix}
    f^2 & 0 \\
    0 & 1
  \end{pmatrix} 
  \qquad \mbox{with} \qquad
  f(s,t) := \sqrt{ \left[ 1 - \frac{t}{R} \cos\left(\frac{s}{2R}\right)\right]^2 
  + \left( \frac{t}{2R} \right)^2 } 
  \,.
\end{equation*}
Without any restriction on the positive parameters~$a$ and~$R$,
the Jacobian~$f$ is always positive, 
and therefore $(\Omega_0,G)$ is an immersed manifold.
In fact, we have the uniform bounds
\begin{equation}\label{bounds}
  \frac{1}{5} \leq f(s,t)^2 \leq
  \left( 1 + \frac{a}{R} \right)^2 
  + \left( \frac{a}{2R} \right)^2
\end{equation}
valid for every $(s,t) \in \Omega_0$.
If one wants to make $(\Omega_0,G)$ embedded 
(and keep the geometric interpretation via a non-overlapping
M\"obius strip~$\Omega$), it is needed to impose the condition $a < R$.
For our purposes, however, it is enough to work in the more
general, immersed setting.

In view of the identification above, 
the Laplace--Beltrami operator $-\Delta_D^\Omega$ in $\sii(\Omega)$
can be identified with the operator 
\begin{equation*}
  H = -|G|^{-1/2} \partial_i |G|^{1/2} G^{ij} \partial_j
  \qquad \mbox{in} \qquad
  \sii(\Omega_0,|G(s,t)|^{1/2} \,\der s \, \der t)
  \,,
\end{equation*} 
subject to Dirichlet boundary conditions~\eqref{bc.Dirichlet}
and twisted periodic boundary conditions \eqref{bc1}--\eqref{bc2}.
Here we use the Einstein summation convention 
with the range of indices $i,j=1,2$,
$|G(s,t)| := \det(G) = f^2$ 
and~$G^{ij}$ denote the coefficients of the inverse metric~$G^{-1}$.
More specifically, $H$~is introduced as the self-adjoint operator 
associated with the closed quadratic form 
\begin{equation*}
\begin{aligned}
  h[\psi] &:= \int_{\Omega_0} 
  \overline{\partial_i\psi(s,t)} \, G^{ij}(s,t) \, \partial_j\psi(s,t)
  \, |G(s,t)|^{1/2} \, \der s \, \der t \,,
  \\
  \Dom(h) &:= \left\{ \psi \in W^{1,2}(\Omega_0) : \
  \mbox{$\psi$ satisfies \eqref{bc.Dirichlet} and \eqref{bc1}}
  \right\} \,.
\end{aligned}
\end{equation*}
Notice that the form domain coincides with the form domain 
of the fake M\"obius strip~$-\Delta_D^{\Omega_0}$ (\cf~\eqref{form.fake})
as well as the not-so-fake model~$H_\mathrm{eff}$.

Since we are interested in the limit $a \to 0$,
it is convenient to introduce the unitary transform 
$U: \sii(\Omega_0,f(s,t)\,\der s \, \der t) \to \sii(\Pi)$ by setting
\begin{equation*}
  (U\psi)(s,u) := \sqrt{a} \, \sqrt{f(s,a u)} \, \psi(s,a u) 
  \,,
\end{equation*}
where~$\Pi$ is the $a$-independent rectangle
\begin{equation}\label{rectangle}
  \Pi := (0,2\pi R) \times (-1,1)
  \,.
\end{equation}
Define the unitarily equivalent operator $L := U H U^{-1}$ in $\sii(\Pi)$,
which is the operator associated with the quadratic form
$l[\phi] := h[U^{-1}\phi]$, $\Dom(l) := U\Dom(h)$.

\begin{Proposition}\label{Prop.form}
One has  
\begin{equation*}
\begin{aligned}
  l[\phi] 
  &= \int_{\Pi} \frac{|\partial_1\phi(s,u)|^2}{f_a(s,u)^2} \, \der s \, \der u
  + \frac{1}{a^2} \int_{\Pi} |\partial_2\phi(s,u)|^2 \, \der s \, \der u
  + \int_{\Pi} V_a(s,u) \, |\phi(s,u)|^2 \, \der s \, \der u
  \,,
  \\
  \Dom(l)
  &= \left\{ \phi \in W^{1,2}(\Pi) : \
  \phi(s,\pm1) = 0 
  \mbox{ for $s \in (0,2 \pi R)$}
  \quad\text{and}\quad
  \phi(0,u) = \phi(2 \pi R,-u)
  \mbox{ for $u \in (-1,1)$}
  \right\} ,
\end{aligned}
\end{equation*}
where $f_a(s,u) := f(s,a u)$ and 
\begin{equation*}\label{potential}
  V_a := -\frac{5}{4} \frac{(\partial_1 f_a)^2}{f_a^4}
  + \frac{1}{2} \frac{\partial_1^2 f_a}{f_a^3}
  - \frac{1}{4} \frac{(\partial_2 f_a)^2}{a^2 f_a^2} 
  + \frac{1}{2} \frac{\partial_2^2 f_a}{a^2 f_a}
  \,.
\end{equation*}
\end{Proposition}
\begin{proof}
The mapping property $U W^{1,2}(\Omega_0) = W^{1,2}(\Pi)$
is easily checked with help of~\eqref{bounds} 
and the fact that there is a constant~$C$
such that $|\nabla f(s,t)| \leq C$  
for every $(s,t) \in \Omega_0$.
If~$\psi$ satisfies the Dirichlet boundary conditions~\eqref{bc.Dirichlet},
then $\phi := U\psi$ clearly satisfies $\phi(s,\pm 1) = 0$ 
for almost every $s \in (0,2\pi R)$.
At the same time, if~$\psi$ satisfies the twisted periodic  
boundary condition~\eqref{bc1},
then~$\phi$ satisfies $\phi(0,u) = \phi(2 \pi R,-u)$
for almost every $u \in (-1,1)$ 
due to the symmetry property $f(2\pi R,a u) = f(0,-a u)$.
These prove the identity for the form domain $\Dom(l)$.

Now, let $\psi \in \Dom(h)$ and $\phi := U\psi$.
Then, making the change of variables and integrating by parts,  
we have
\begin{equation*}
\begin{aligned}
  h[\psi] 
  &= \int_{\Pi} \overline{\partial_i\phi} \, G_a^{ij} \, \partial_j\phi
  - \int_{\Pi} \frac{1}{2} \frac{\partial_i f_a}{f_a} G_a^{ij} \partial_j|\phi|^2
  + \int_{\Pi} \frac{1}{4}  
  \frac{\partial_i f_a}{f_a} G_a^{ij} \frac{\partial_i f_a}{f_a} |\phi|^2
  \\
  &= \int_{\Pi} \overline{\partial_i\phi} \, G_a^{ij} \, \partial_j\phi
  + \int_{\Pi} 
  \left[
  \frac{1}{2} \partial_j\!\left(\frac{\partial_i f_a}{f_a} G_a^{ij}\right)  
  + \frac{1}{4} \frac{\partial_i f_a}{f_a} G_a^{ij} \frac{\partial_i f_a}{f_a}
  \right]
  |\phi|^2
  \,,
\end{aligned}
\end{equation*}
where $G_a(s,u) := \diag(f_a^2,a^2)$
and the arguments $(s,u) \in \Pi$ of the functions 
and the measure of integration $\der s \, \der u$ are suppressed.
Here it is crucial that the boundary terms due to the integration by parts vanish.
This is easy to see when we integrate by parts with respect to the second variable,
because of the Dirichlet boundary conditions $\phi(s,\pm 1) = 0$
for almost every $s \in (0,2\pi R)$.
To see it also when we integrate by parts with respect to the first variable, 
we notice the property $\partial_1 f_a(0,u) = 0 = \partial_1 f_a(2\pi R,u)$
for every $u \in (-1,1)$.
Using the special form of~$G_a$, we get the desired formula.
\end{proof}

We observe that there exists an $a$-independent constant~$C$ such that
\begin{equation}\label{f.bounds}
  |f_a(s,u) - 1| \leq C a 
\end{equation}
and
\begin{equation*} 
\begin{aligned}
  |\partial_1 f_a(s,u)| &\leq C a \,,
  &&& |\partial_1^2 f_a(s,u)| &\leq C a^2 \,,
  \\
  \left|
  \frac{\partial_2 f_a(s,u)}{a} + \frac{1}{R} \cos\left(\frac{s}{2R}\right)  
  \right| &\leq C a \,,	
  &&&\left|
  \frac{\partial_2^2 f_a(s,u)}{a^2} - \frac{1}{(2R)^2}  
  \right| &\leq C a \,,
\end{aligned}
\end{equation*}
for every $(s,u) \in \Pi$.
Hereafter we use the convention that~$C$ 
denotes an $a$-independent constant
which may change its value from one line to another.
Consequently, there exists another constant~$C$ such that
\begin{equation}\label{potential.limit}
  |V_a(s, u) - V_\mathrm{eff}(s, au)| \leq C a
\end{equation}
for every $(s,u) \in \Pi$, where~$V_\mathrm{eff}$ is defined in~\eqref{eff}.

\begin{Remark}\label{Rem.Fermi}
Since~$\Omega$ is a ruled surface,
the straight lines $t \mapsto \mathscr{L}(s,t)$ are geodesics
for every fixed $s \in [0,2 \pi R)$.
Consequently, $(s,t)$ are Fermi coordinates 
and we have the simple formula 
\begin{equation*}
  K := - \frac{\partial_2^2 f}{f}
\end{equation*}
for the Gauss curvature of~$\Omega$.
At the same time, using the normal vector field $\partial_2\mathscr{L}$
(actually independent of the second variable)
along the unit-speed circle $s \mapsto \mathscr{L}(s,0)$ as a curve on~$\Omega$,
we have the formula
\begin{equation*}
  \kappa_g(s) 
  = \partial_1\mathscr{L}(s,0) \cdot \partial_2\mathscr{L}(s,0)
  = \frac{1}{R} \cos\left(\frac{s}{2R}\right)
\end{equation*}
for the geodesic curvature of the circle. 
Consequently,
\begin{equation*}
  V_\mathrm{eff}(s,t) =-\frac{1}{4} \kappa_g(s)^2 - \frac{1}{2} K(s,0)
  \,,
\end{equation*}
in agreement with the case of general thin strips \cite[Sec.~3]{K1}.
Notice that while the geodesic curvature has a jump 
at the endpoints of the M\"obius strip,
namely $\kappa_g(0) = - \kappa_g(2 \pi R)$,
which reflects the fact that~$\Omega$ is not orientable, 
the effective potential~$V_\mathrm{eff}$ extends 
to a \emph{smooth} function on~$\Omega$.
\end{Remark}

In order to compare the true M\"obius strip~$H$
(which is unitarily equivalent to $L$ in $\sii(\Pi)$)
with the not-so-fake model~$H_\mathrm{eff}$ in $\sii(\Omega_0)$,
we also map the latter to an operator 
in the $a$-independent Hilbert space $\sii(\Pi)$.
This is achieved by the unitary transform
$U_\mathrm{eff}: \sii(\Omega_0) \to \sii(\Pi)$ that acts as
\begin{equation*}
  (U_\mathrm{eff}\psi)(s,u) := \sqrt{a} \, \psi(s,a u) 
  \,.
\end{equation*}
It is elementary to check that the unitarily equivalent operator
$L_\mathrm{eff} := U_\mathrm{eff} H_\mathrm{eff} (U_\mathrm{eff})^{-1}$
is associated with the quadratic form 
\begin{equation*}
\begin{aligned}
  l_\mathrm{eff}[\phi] 
  &:= \int_{\Pi} |\partial_1\phi(s,u)|^2 \, \der s \, \der u
  + \frac{1}{a^2} \int_{\Pi} |\partial_2\phi(s,u)|^2 \, \der s \, \der u
  + \int_{\Pi} V_\mathrm{eff}(s, au) \, |\phi(s,u)|^2 \, \der s \, \der u
  \,,
  \\
  \Dom(l_\mathrm{eff})
  &:= \Dom(l)
  \,,
\end{aligned}
\end{equation*}
where $\Dom(l)$ is given in Proposition~\ref{Prop.form}.

Now we are in a position to establish the norm-resolvent convergence.

\begin{Theorem}\label{Thm.nrs}
For every $z \not\in \sigma(L-E_1) \cup \sigma(L_\mathrm{eff}-E_1)$,
there exists an $a$-independent constant~$C$ such that,
for all positive~$a$,
\begin{equation}\label{nrs.thm}
  \big\|(L-E_1-z)^{-1} 
  - (L_\mathrm{eff}-E_1-z)^{-1} \big\| \leq C a
  \,,
\end{equation}
where $\|\cdot\|$ denotes the operator norm in $\sii(\Pi)$.
\end{Theorem}
\begin{proof}
Using~\eqref{Poincare} and an elementary estimate of~$V_\mathrm{eff}$, 
we have 
\begin{equation*}
  L_\mathrm{eff} -E_1 \geq - \frac{1}{8R^2}
  \,. 
\end{equation*}
Consequently, any negative~$z$ with sufficiently large~$|z|$ 
(with the largeness independent of~$a$)
belongs to the resolvent set of~$H_\mathrm{eff}$.
Using~\eqref{potential.limit}, the same conclusion holds for~$L$. 
Fixing such an $a$-independent~$z$ 
and given arbitrary functions $f,g \in \sii(\Pi)$, 
let us consider the resolvent equations
\begin{equation}\label{res.eqs}
  (L_\mathrm{eff}-E_1-z)\phi = f
  \qquad \mbox{and} \qquad
  (L-E_1-z)\psi = g
  \,.
\end{equation}
The first equation implies 
\begin{equation*}
  l_\mathrm{eff}[\phi] -E_1 \|\phi\|^2 -z\|\phi\|^2 
  = (\phi,f)
  \leq \|\phi\| \|f\|
  \,,
\end{equation*}
where~$\|\cdot\|$ and $(\cdot,\cdot)$ denote 
the norm and inner product of $\sii(\Pi)$.
Recalling~\eqref{Poincare}, we obtain
\begin{equation}\label{estimates1}
  \|\phi\| \leq C \|f\|
  \qquad\mbox{and}\qquad
  \|\partial_1\phi\|^2 \leq C \|f\|^2,
\end{equation}
where $C := [|z|-(8R^2)^{-1}]^{-1}$.
Similarly, using in addition~\eqref{f.bounds} and~\eqref{potential.limit},
the second equation of~\eqref{res.eqs} yields
\begin{equation}\label{estimates2}
  \|\psi\| \leq C \|g\|
  \qquad\mbox{and}\qquad
  \|\partial_1\psi\|^2 \leq C \|f\|^2
\end{equation}
with some $a$-independent constant~$C$. 
Let us now write
\begin{align}\label{diff.res} 
  \big(f,[(L-E_1-z)^{-1}
  -(L_\mathrm{eff}-E_1-z)^{-1}]g\big)
  &= (f,(L-E_1-z)^{-1}g) 
  - ((L_\mathrm{eff}-E_1-z)^{-1}f,g) 
  \nonumber \\
  &= \big((L_\mathrm{eff}-E_1-z)\phi,\psi\big) 
  - \big(\phi,(L-E_1-z)\psi\big)
  \nonumber \\ 
  &= l_\mathrm{eff}(\phi,\psi) - l(\phi,\psi) 
  \,,
\end{align}
where $l(\cdot,\cdot)$ (respectively, $l_\mathrm{eff}(\cdot,\cdot)$)
denotes the sesquilinear form associated with $l[\cdot]$
(respectively, $l_\mathrm{eff}[\cdot]$). 
The last identity employs the fact that the form domains 
of~$L$ and $L_\mathrm{eff}$ coincide.
We have
\begin{align}\label{diff.forms} 
  |l(\phi,\psi) - l_\mathrm{eff}(\phi,\psi)|
  &= \left|
  \big(\partial_1\phi,[f_a^{-2}-1]\partial_1\psi\big) 
  + \big(\phi,[V_a-V_\mathrm{eff}]\psi\big) 
  \right|
  \nonumber \\ 
  &\leq C a \|\partial_1\phi\| \|\partial_1\psi\|
  + C a \|\phi\| \|\psi\|
  \nonumber \\
  &\leq C a \|f\| \|g\|  
  \,,
\end{align}  
where the first estimate follows by~\eqref{f.bounds} and~\eqref{potential.limit}
together with the Schwarz inequality
and the second inequality employs~\eqref{estimates1} and~\eqref{estimates2}.
Combining~\eqref{diff.res} and~\eqref{diff.forms},
we obtain~\eqref{nrs.thm}. 
In view of~ \cite[Rem.~IV.3.13]{Kato},
the estimate~\eqref{nrs.thm} extends to any~$z$ in the resolvent
sets of~$L$ and~$L_\mathrm{eff}$. 
\end{proof}
As a particular consequence of~\eqref{nrs.thm} and~\cite[Chapter II, Corrolary 2.3]{GohbergKrein},
we get the convergence of eigenvalues of~$L$ 
to the eigenvalues of~$L_\mathrm{eff}$.
More specifically, 
for fixed real $z \not\in \sigma(L) \cup \sigma(L_\mathrm{eff})$
and every $j \in \Nat^*$,
we have
\begin{equation*}
  \left|
  \frac{1}{\lambda_j(L)-E_1-z} - \frac{1}{\lambda_j(L_\mathrm{eff})-E_1-z} 
  \right|
  \leq C a
  \,,
\end{equation*}
where the $a$-independent constant~$C$ is the same as in~\eqref{nrs.thm}
and $\{\lambda_j(A)\}_{j\in\Nat^*}$ denotes 
the non-decreasing sequence of eigenvalues 
of a self-adjoint operator~$A$ with compact resolvent,
where each eigenvalue is repeated according to its multiplicity
(\cf~\cite[Sec.~4.5]{Davies}).
The eigenvalues $\lambda_j(L_\mathrm{eff})$ are known explicitly, see~\eqref{spec.eff}.
In particular, given any $j \in \Nat^*$,
the shifted eigenvalue $\lambda_j(L_\mathrm{eff})-E_1$ is independent of~$a$
for all sufficiently small~$a$ (with the smallness depending on~$j$)
and we thus get the following result.
\begin{Corollary}\label{Cor.eigenvalues}
For every $k \in \Nat^*$, 
there exist positive $a$-independent constants~$C_k$ and~$a_k$ such that,
for all $a \leq a_k$,
$$
  \forall j \in \{1,\dots,k\} \,, \qquad
  |\lambda_j(L) - \lambda_j(L_\mathrm{eff})| \leq C_k \, a 
  \,.
$$
\end{Corollary}
The convergence in norm of corresponding spectral projections
also follows.

%--------------------------%
\section{Numerical results}
\label{Sec.numerics}
%--------------------------%
%
In this closing section, we will numerically investigate properties of the eigenvalues and eigenfunctions of the true and not-so-fake M\"obius models described above.
In view of the unitary equivalence described in Section~\ref{Sec.nrs},
our starting point is the operator $L$ associated with the quadratic form~$l$ 
introduced in Proposition~\ref{Prop.form}.
It corresponds to the operator
\begin{equation*}\label{eq.L}
  L = - \partial_1 \frac{1}{f_a^2} \partial_1 - \frac{1}{a^2} \partial_2^2 + V_a
  \qquad\mbox{in}\qquad 
  \sii(\Pi)
  \,, 
\end{equation*}
subject to Dirichlet boundary conditions 
\begin{equation}\label{bc.Dirichlet.num}
  \psi(s,\pm 1) = 0 
\end{equation}
for almost every $s \in (0,2\pi R)$
and twisted periodic boundary conditions
\begin{align}
  \psi(0,u) &= \psi(2 \pi R,-u) \,, 
  \label{bc1.num}
  \\ 
  \partial_1\psi(0,u) &= \partial_1\psi(2 \pi R,-u) \,,
  \label{bc2.num}
\end{align}
for almost every $u \in (-1,1)$.
The rectangular domain is defined in~\eqref{rectangle}
and the functions~$V_a$ and~$f_a$ can be found in Proposition~\ref{Prop.form}.

In order to numerically analyse solutions of the eigenvalue problem $L f = \lambda f$, 
we employ a particular orthonormal basis of $L^2(\Pi)$ 
formed by eigenfunctions of the fake model,
\[
  \mathcal{B} := \{\Psi_{m,n}\}_{\substack{m\in\Int,n\in\Nat^* \\ m+n \ \text{is odd}}} \subset L^2(\Pi),
\]
where $\Psi_{m,n}(s,u) := \sqrt{a} \, \psi_{m,n}(s,au)$ 
and $\psi_{m,n}$ is defined in~\eqref{eq.psimn}.
For convenience, let us arrange 
the eigenvalues~\eqref{spec.fake} of the fake model
in a non-decreasing sequence $\{ \lambda_j^{(\mathrm{fake})} \}_{j\in\Nat^*}$,
where each eigenvalue is repeated according to its multiplicity.
The set of corresponding eigenfunctions will be denoted by 
$\mathcal{B} = \{ \Psi_j \}_{j\in\Nat^*}$.
Thus 
$$
  \left(- \partial_1^2 - \frac{1}{a^2} \partial_2^2\right) \Psi_j
  = \lambda_j^{(\mathrm{fake})} \Psi_j 
$$
and $\Psi_j$ with $j \in \Nat^*$ obey
the Dirichlet~\eqref{bc.Dirichlet.num}
as well as twisted periodic boundary conditions~\eqref{bc1.num}--\eqref{bc2.num}.
Now let us fix a large $N \in \Nat^*$ and consider the orthogonal projection $P_N$ onto the linear span of the truncated orthonormal set $\mathcal{B}_N = \{ \Psi_j \}_{j=1}^N$.
Instead of the full eigenvalue problem for $L$ 
we solve the finite-dimensional eigenvalue problem for $P_N L P_N$.
In other words, 
we compute eigenvalues and eigenvectors of the matrix $M \in \mathbb{R}^{N,N}$ 
with entries
\[
  M_{jk} = ( \Psi_j, L \Psi_k ), 
  \qquad j,k = 1,\ldots,N \,,
\]
where $(\cdot,\cdot)$ denotes the inner product in~$\sii(\Pi)$.
The resulting eigenvalues 
$\tilde{\lambda}^{(\mathrm{true})}_1 \leq \tilde{\lambda}^{(\mathrm{true})}_2 \leq \cdots \leq \tilde{\lambda}^{(\mathrm{true})}_N$ are upper estimates 
(\cf~\cite[Sec.~4.5]{Davies})
of the true M\"obius eigenvalues, \ie{} we have $\lambda^{(\mathrm{true})}_j \leq \tilde{\lambda}^{(\mathrm{true})}_j$, $j=1,2,\ldots,N$.
Corresponding approximations to the true eigenvectors are then
\[
  \tilde{f}_k = \sum_{j=1}^N c_j^{(k)} \Psi_j,
\]
where $c^{(k)} = (c^{(k)}_1,\ldots,c^{(k)}_N)^T$ is an eigenvector of $M$ with eigenvalue $\tilde{\lambda}^{(\mathrm{true})}_k$.
Note that the entries of $M$ have to be evaluated numerically, in particular we have the following expression
\begin{align*}
  \langle \Psi_{m,n}, L \Psi_{k,\ell} \rangle =\ &
  \frac{mk}{4R^2} \int_\Pi \frac{1}{f_a(s,u)^2} \Psi_{-m,n}(s,u) \Psi_{-k,\ell}(s,u) 
  \,\mathrm{d}s\,\mathrm{d}u +
  \frac{1}{a^2} \left( \frac{n\pi}{2} \right)^2 \delta_{mk}\delta_{n\ell} + \\
  &+\int_\Pi V_a(s,u) \Psi_{m,n}(s,u) \Psi_{k,\ell}(s,u) \,\mathrm{d}s\,\mathrm{d}u,
\end{align*}
where $\delta_{mn}$ is the usual Kronecker symbol.
\begin{Remark}
  Before we conclude this section with the presentation of our numerical results, let us  make two remarks.  
  The matrix $M$ is a $N \times N$ real symmetric and full matrix (\ie{} not a sparse one).
  Its eigenvalues and eigenvectors can be computed numerically using one's favorite computer algebra system, we use the Julia programming environment \cite{Julia}.
  All of our code is available in a public GitHub repository\footnote{\url{https://github.com/kalvotom/moebius}}.
  The reason for taking the eigenvectors of the fake model instead of not-so-fake model is that it is much easier to work with sines and cosines instead of Mathieu functions. 
\end{Remark}
Let us conclude this section with a short review of the results of our numerical experiments.
In Figure \ref{fig.eigenvectors.1} one can see the eigenfunctions of the operator $L$ with $a = 1.3$ and $R = 18/(2\pi)$.
It is also interesting to visualize eigenvectors as living on the original M\"{o}bius strip.
This is the purpose of Figure \ref{fig.eigenvectors.3d.1} where we have taken $a = 0.75$ and $R = 13.2/(2\pi)$.
Table~\ref{tab.eigenvalues} contains approximations of the corresponding first twenty eigenvalues of the operator $L$.
Finally, we wish to test the expected asymptotic behavior of the not-so-fake and true model as the thickness of the M\"obius strip $a$ goes to zero (recall Theorem~\ref{Thm.nrs}).
Let $\lambda_n^{(\textrm{not-so-fake})}(a)$
and $\lambda_n^{(\mathrm{true})}(a)$ denote the $n$th eigenvalue of 
the operator $L_{\mathrm{eff}}$ and $L$, respectively.
The former is known explicitly (see equation~\eqref{spec.eff}), the latter can be computed numerically as described in the previous paragraphs.
The resulting approximation well be denoted by $\tilde{\lambda}_n^{(\mathrm{true})}$.

In  order to test the convergence rate we plot the ratio
\[
  \frac{\big| \lambda_n^{(\textrm{not-so-fake})}
  - \tilde{\lambda}_n^{(\textrm{true})} \big|}{a^2}
\]
for $n = 1,\ldots,20$ and $a$ ranging from $0.01$ to $1.5$ with $R = 18/(2\pi)$, see Figure~\ref{fig.diff.ratio}.
In this particular case we also check convergence of the corresponding normalized eigenvectors.
In Figure~\ref{fig.eigenvectors.convergence} we plot the ratio
\[
  \frac{\big\|\tilde{f}_n - f_n^{(\textrm{not-so-fake})}\big\|_{L^2(\Pi)}}{a^2}.
\]
Both of these experiments suggest that the convergence rate is quadratic as $a$ goes to $0$ and so the Corollary~\ref{Cor.eigenvalues} is only a good upper estimate.
%

%---------------------------%
\subsection*{Acknowledgment}
%---------------------------%
%
The research of D.K.\ was partially supported 
by the GACR grant No.\ 18-08835S.
The research of T.K.\ was supported by the Ministry of Education, Youth and Sports of the Czech Republic project no. CZ.02.1.01/0.0/0.0/16\_019/0000778.

%\newpage
%\vfill
%--------------%
% BIBLIOGRAPHY %
%--------------%
%
%\addcontentsline{toc}{section}{References}
% \bibliography{bib}
\bibliographystyle{amsplain}
\providecommand{\bysame}{\leavevmode\hbox to3em{\hrulefill}\thinspace}
\providecommand{\MR}{\relax\ifhmode\unskip\space\fi MR }
% \MRhref is called by the amsart/book/proc definition of \MR.
\providecommand{\MRhref}[2]{%
  \href{http://www.ams.org/mathscinet-getitem?mr=#1}{#2}
}
\providecommand{\href}[2]{#2}

\newpage

\begin{table}
  {\centering
  \begin{tabular}{lll|ll}
    $n$ & $\tilde{\lambda}^{(\mathrm{true})}_n$ & $\big\| L \tilde{f}_n - \tilde{\lambda}^{(\mathrm{true})}_n \tilde{f}_n \big\|_{L^2(\Pi)}$ & $\lambda^{(\mathrm{not-so-fake})}_n$ & $\lambda^{(\mathrm{fake})}_n$ \\ \hline

1 & 4.387440201465426 & 0.0011360336639659758 & 4.384732657634105 & 4.386490844928603\\
2 & 4.619975308169118 & 0.002935713704540701 & 4.612770845791257 & 4.613065785265825\\
3 & 4.6210487512326965 & 0.0034765508104058836 & 4.61452884109396 & 4.613065785265825\\
4 & 5.311812674844678 & 0.009392208389967776 & 5.292908532928366 & 5.292790606277488\\
5 & 5.311812691949888 & 0.009394620959796087 & 5.292908724918286 & 5.292790606277488\\
6 & 6.45928381512197 & 0.01784426849679324 & 6.425715883668621 & 6.425665307963595\\
7 & 6.459283815177474 & 0.017844019253709244 & 6.425715883670497 & 6.425665307963595\\
8 & 8.054793717112888 & 0.02782741553208048 & 8.011717987565671 & 8.011689890324144\\
9 & 8.054793717134626 & 0.02782929178936725 & 8.011717987565671 & 8.011689890324144\\
10 & 10.087710686170643 & 0.07623428743234176 & 10.050882233363655 & 10.050864353359135\\
11 & 10.087710686180136 & 0.07623425520146826 & 10.050882233363655 & 10.050864353359135\\
12 & 12.544971054834159 & 0.12299616532523619 & 12.543201075519463 & 12.543188697068569\\
13 & 12.544971054880232 & 0.12299618315689324 & 12.543201075519463 & 12.543188697068569\\
14 & 15.411764278613166 & 0.15141422162634224 & 15.48867199897889 & 15.488662921452445\\
15 & 15.411764278618152 & 0.1514142253244023 & 15.48867199897889 & 15.488662921452445\\
16 & 17.59842628782262 & 0.005712405600055002 & 17.58801732145255 & 17.602607114798715\\
17 & 17.622050913758347 & 0.003779453318856355 & 17.616311563972907 & 17.602607114798715\\
18 & 18.084500866091076 & 0.017503712257836673 & 18.055964587231244 & 18.05575699547316\\
19 & 18.084502386722757 & 0.01752620443224552 & 18.055992211502907 & 18.05575699547316\\
20 & 18.672740544194298 & 0.18297451968432338 & 18.887293968147098 & 18.88728702651077

  \end{tabular}
  \par}
  \caption{
    First twenty approximations of the eigenvalues of the true model $\tilde{\lambda}^{(\mathrm{true})}_n$ and norms of residues $L \tilde{f}_n - \tilde{\lambda}^{(\mathrm{true})}_n \tilde{f}_n$, where $\tilde{f}_n$ is the normalised eigenvector of $M$ corresponding to $\tilde{\lambda}^{(\mathrm{true})}_n$.
    Value of parameters are $a = 0.75$, $R = 13.2/(2\pi)$, and $N = 82$.
    For convenience we also present eigenvalues of the not-so-fake $\lambda^{(\mathrm{not-so-fake})}_n$ and fake model $\lambda^{(\mathrm{fake})}_n$, see~\eqref{spec.eff} and~\eqref{spec.fake}, respectively.
  }
  \label{tab.eigenvalues}
\end{table}
\begin{figure}
  \ContinuedFloat*
  {\centering
    \includegraphics[scale=0.7]{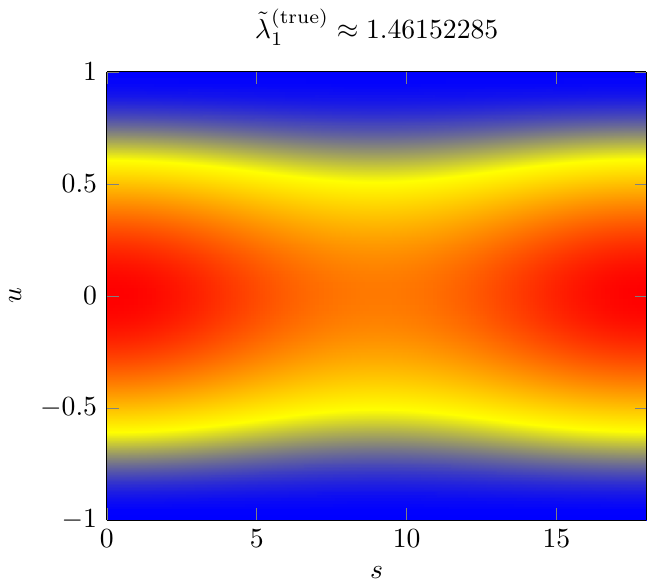}
    \hfill
    \includegraphics[scale=0.7]{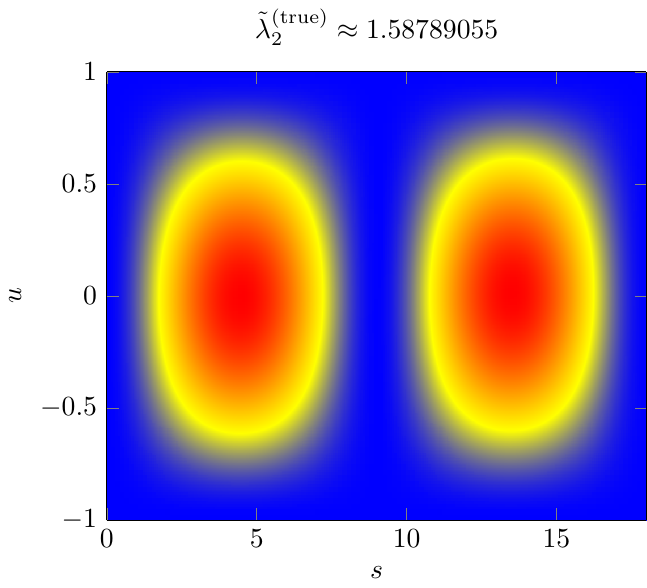}
    \hfill
    \includegraphics[scale=0.7]{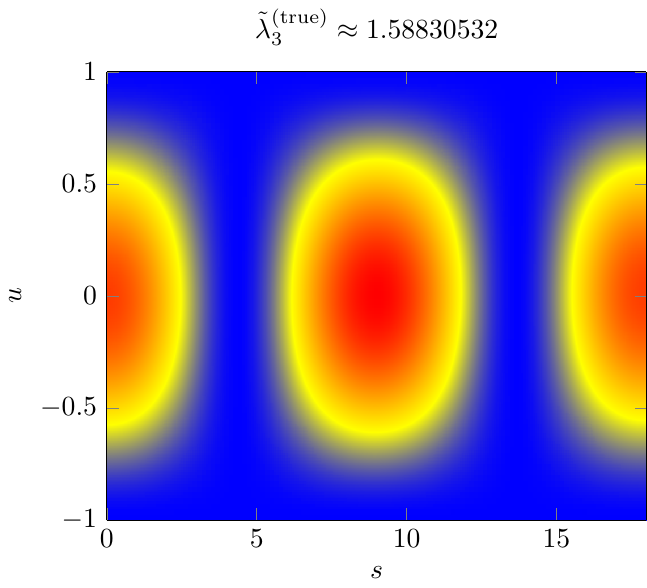}
    \par\medskip
    \includegraphics[scale=0.7]{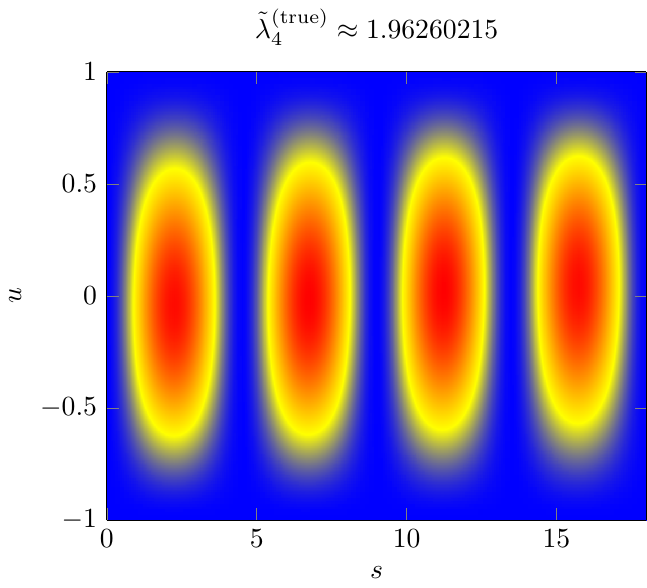}
    \hfill
    \includegraphics[scale=0.7]{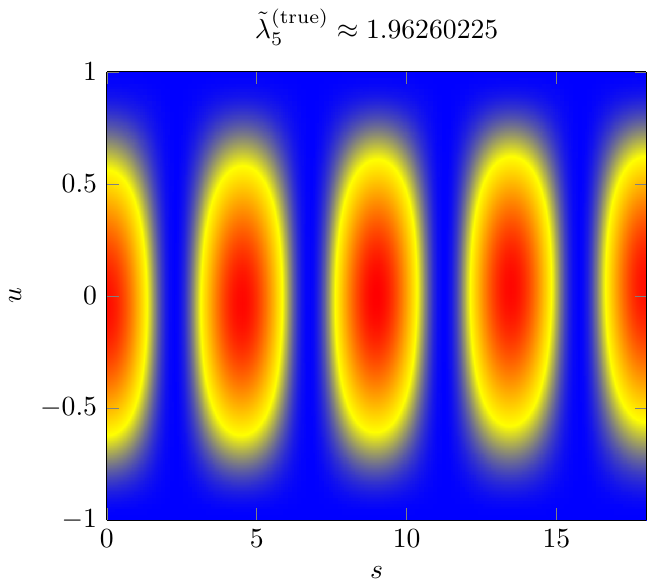}
    \hfill
    \includegraphics[scale=0.7]{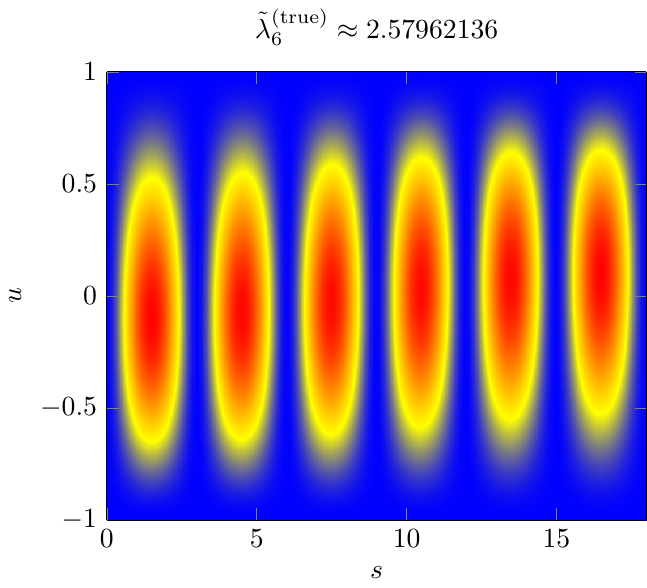}
    \par\medskip
    \includegraphics[scale=0.7]{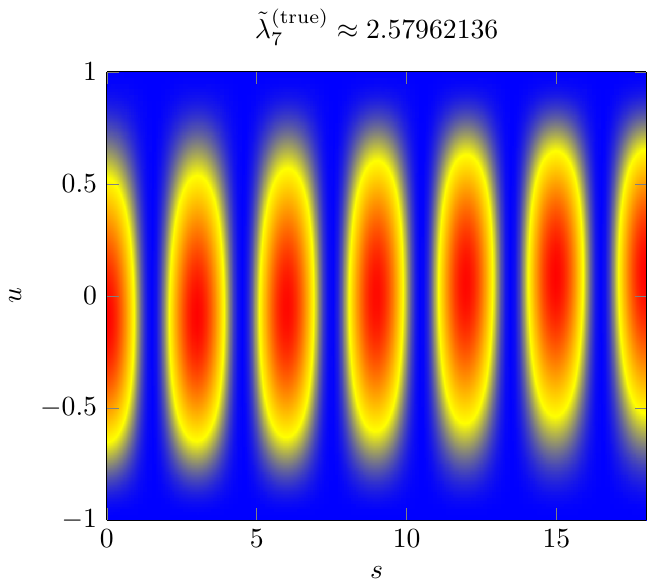}
    \hfill
    \includegraphics[scale=0.7]{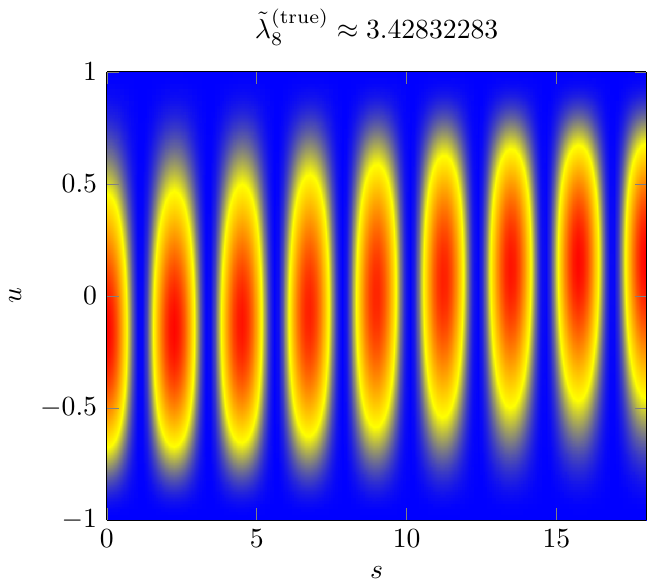}
    \hfill
    \includegraphics[scale=0.7]{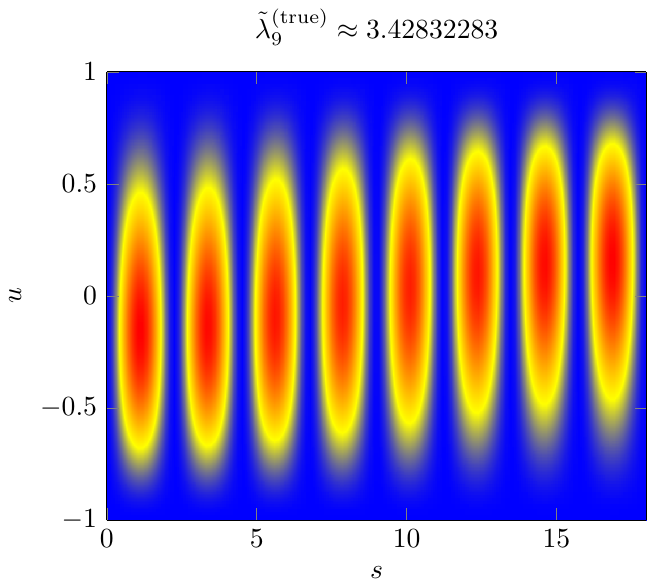}
    \par\medskip
    \includegraphics[scale=0.7]{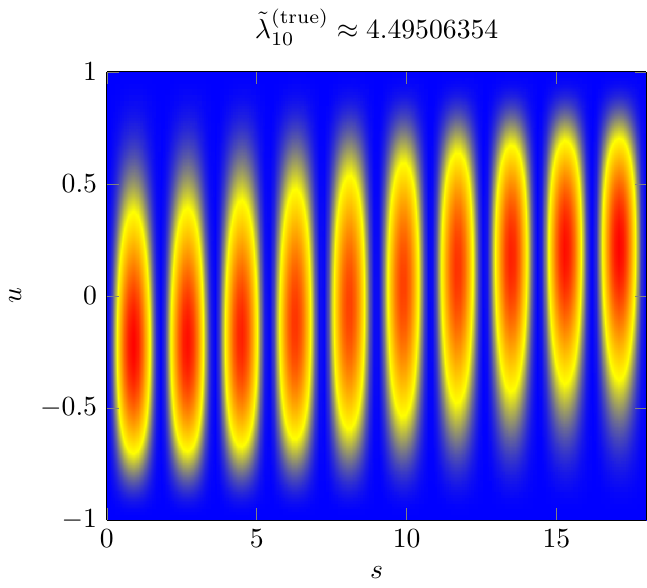}
  \par}
  \caption{
    Numerical approximations $\tilde{f}_k$ of eigenvalues and eigenfunctions of the operator~$L$ for $k=1,2,\ldots,10$.
    We are plotting probability densities $|\tilde{f}_k|^2$, blue and red color corresponds to zero and maximal value, respectively.
    Parameters of the numerical computation are $a = 1.3$, $R = 18/(2\pi)$, and $N = 96$.
  }
  \label{fig.eigenvectors.1}
\end{figure}
\begin{figure}
  \ContinuedFloat
  {\centering
    \includegraphics[scale=0.7]{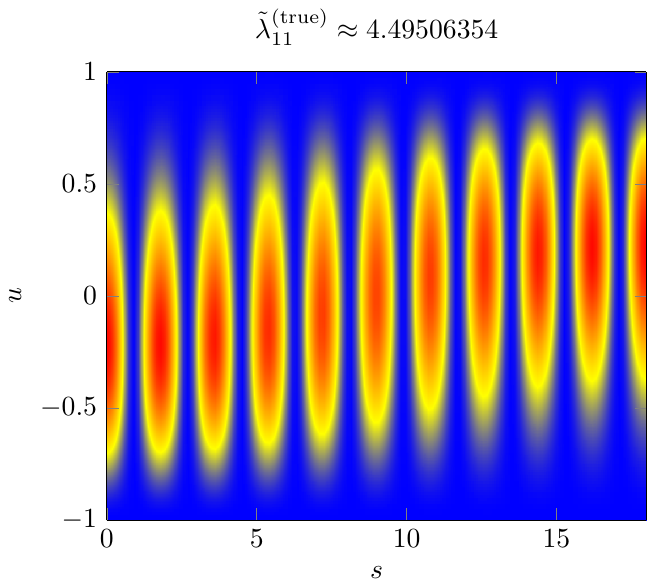}
    \hfill
    \includegraphics[scale=0.7]{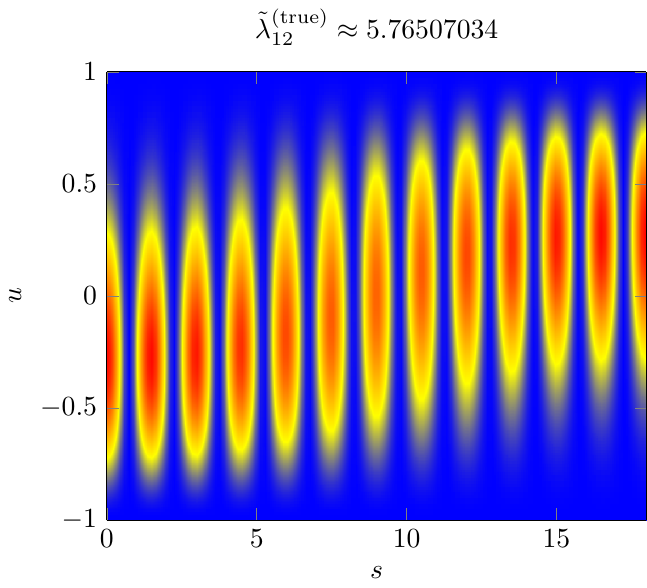}
    \hfill
    \includegraphics[scale=0.7]{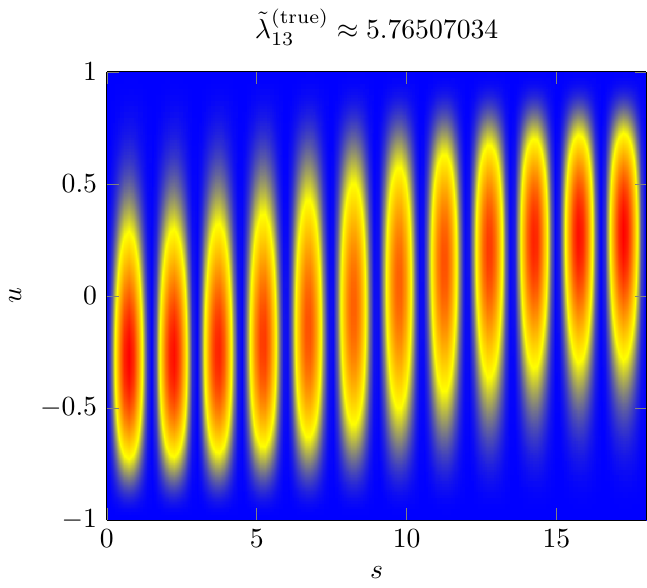}
    \par\medskip
    \includegraphics[scale=0.7]{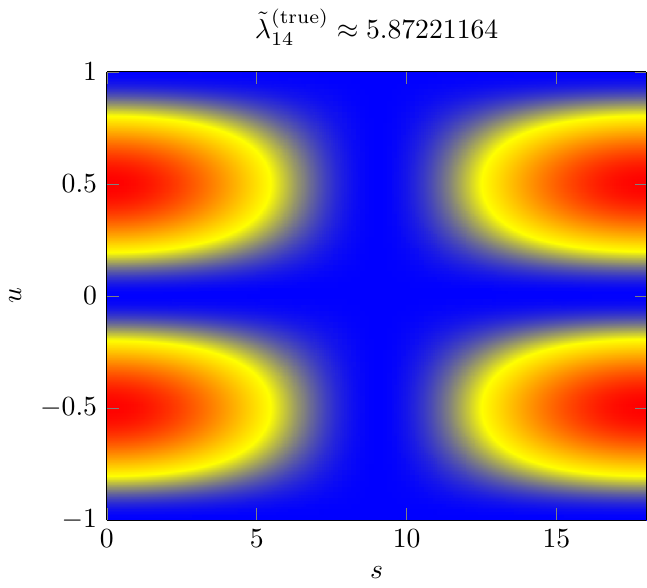}
    \hfill
    \includegraphics[scale=0.7]{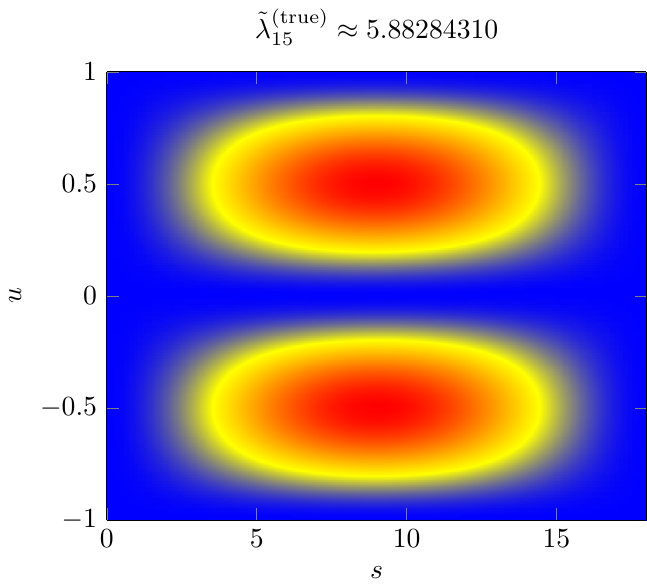}
    \hfill
    \includegraphics[scale=0.7]{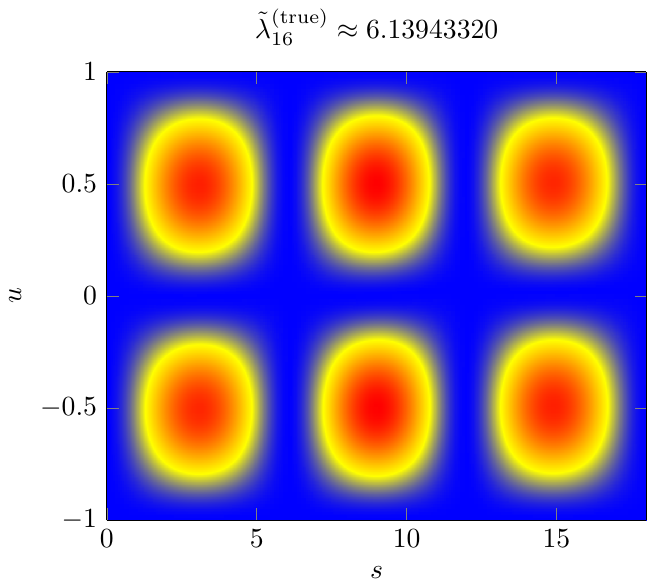}
    \par\medskip
    \includegraphics[scale=0.7]{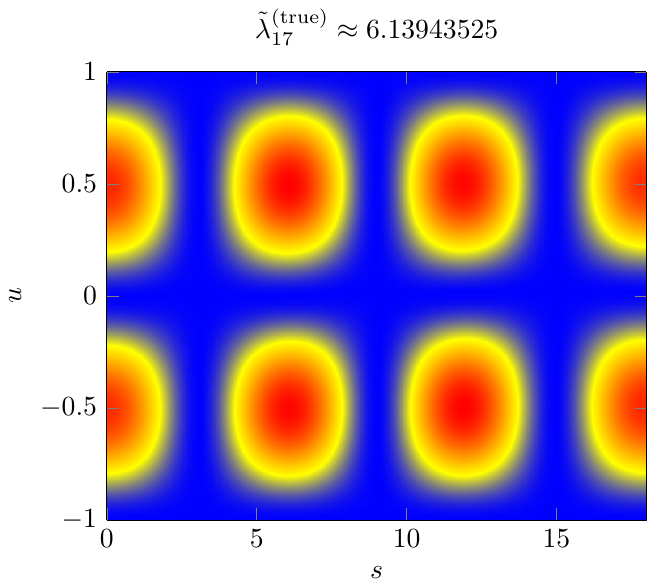}
    \hfill
    \includegraphics[scale=0.7]{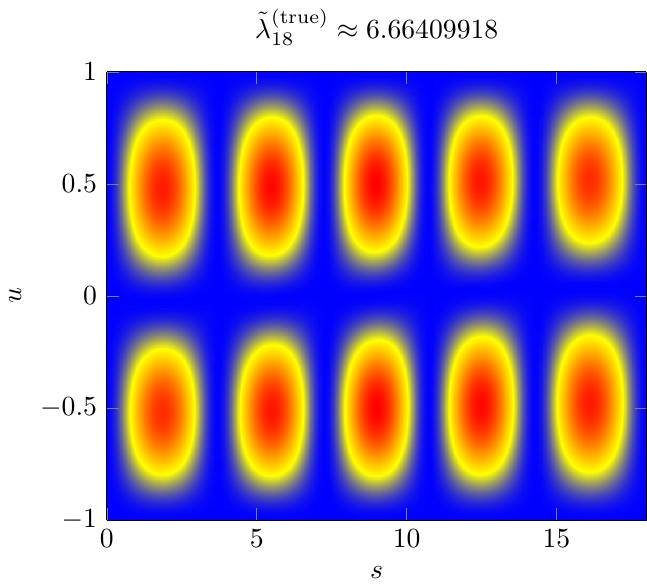}
    \hfill
    \includegraphics[scale=0.7]{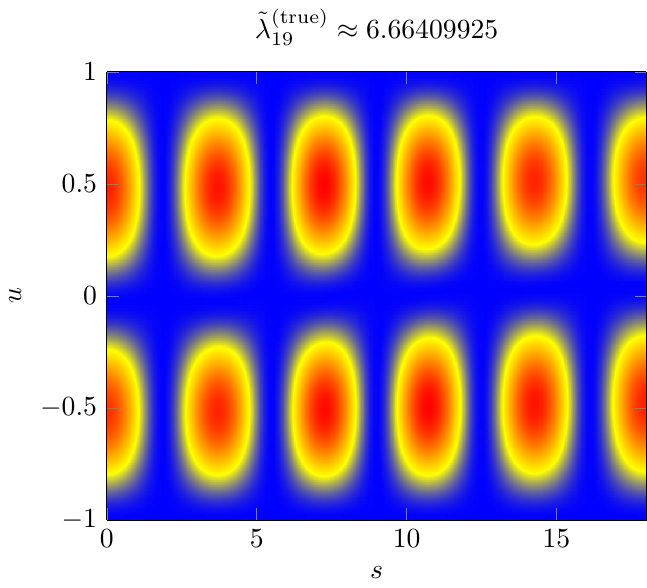}
    \par\medskip
    \includegraphics[scale=0.7]{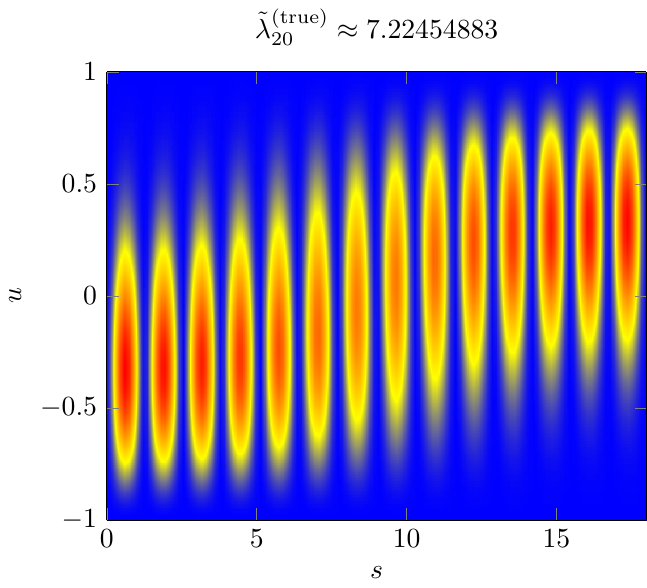}
  \par}
  \caption{
    Numerical approximations $\tilde{f}_k$ of eigenvalues and eigenfunctions of the operator~$L$ for $k=11,12,\ldots,20$.
    We are plotting probability densities $|\tilde{f}_k|^2$, blue and red color corresponds to zero and maximal value, respectively.
    Parameters of the numerical computation are $a = 1.3$, $R = 18/(2\pi)$, and $N = 96$.
  }
  \label{fig.eigenvectors.2}
\end{figure}
\begin{figure}
  \ContinuedFloat*
  {\centering
    \includegraphics[scale=0.7]{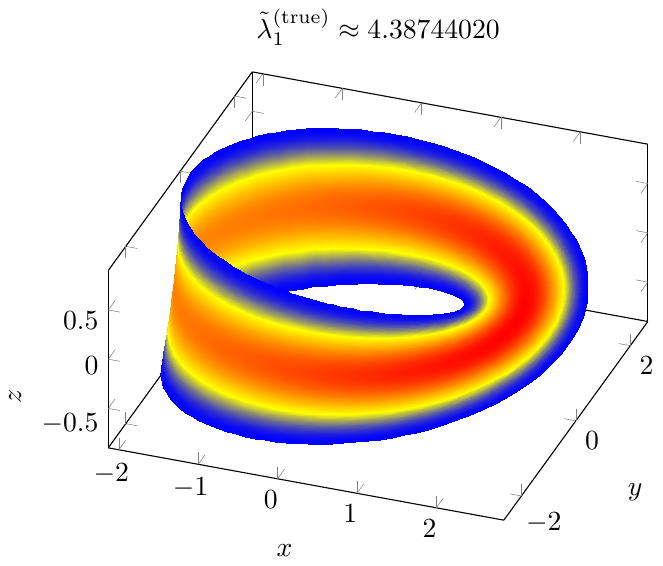}
    \hfill
    \includegraphics[scale=0.7]{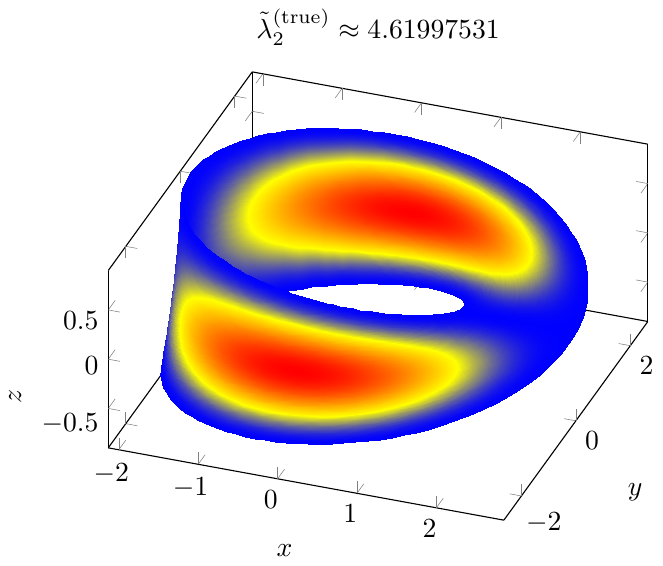}
    \hfill
    \includegraphics[scale=0.7]{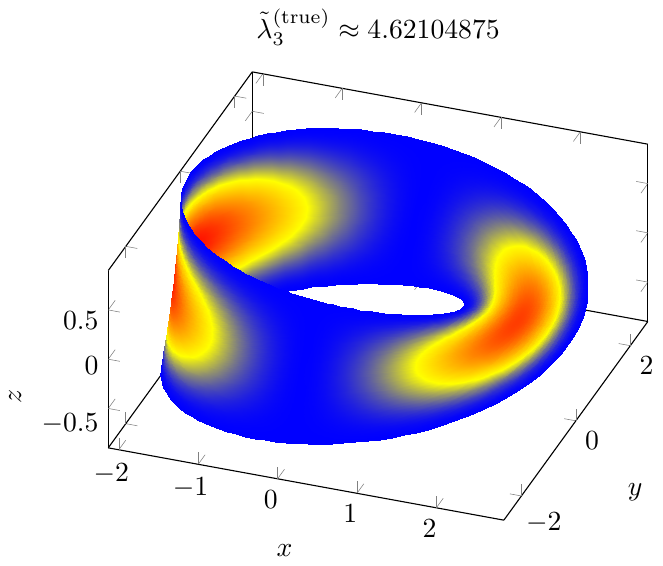}
    \par\medskip
    \includegraphics[scale=0.7]{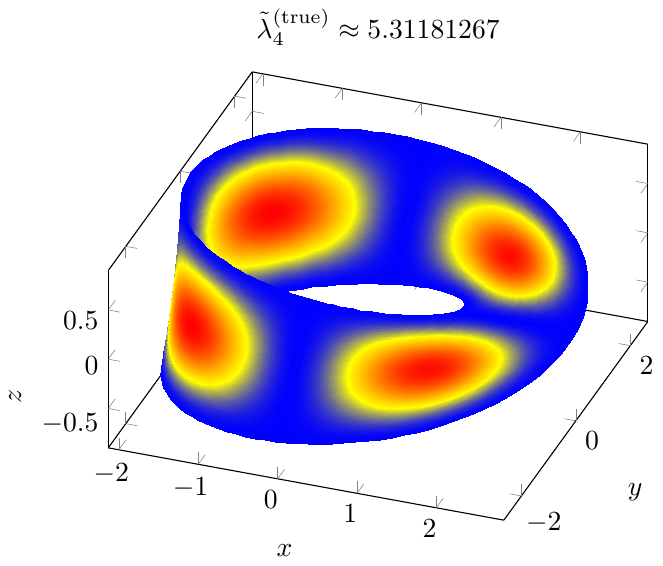}
    \hfill
    \includegraphics[scale=0.7]{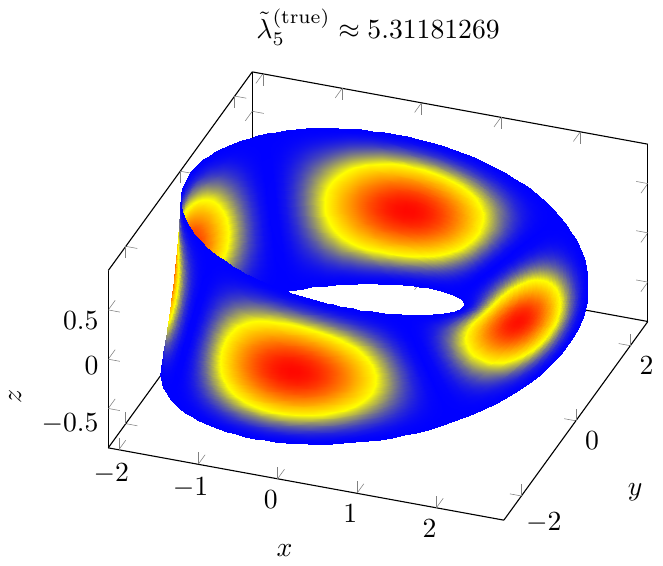}
    \hfill
    \includegraphics[scale=0.7]{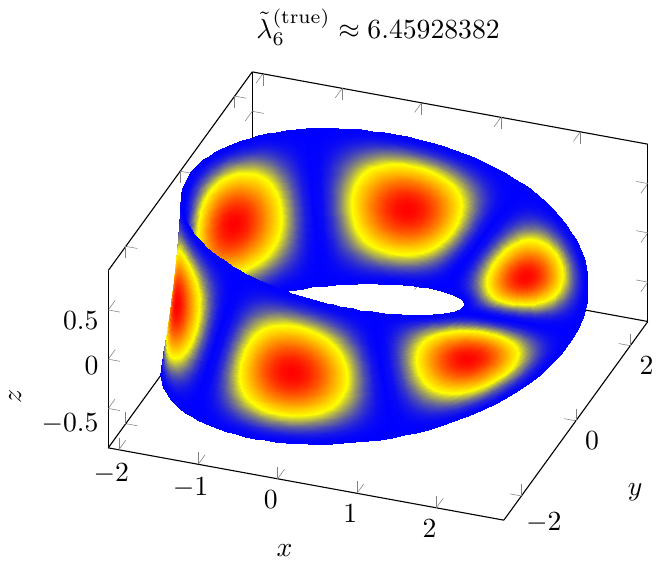}
    \par\medskip
    \includegraphics[scale=0.7]{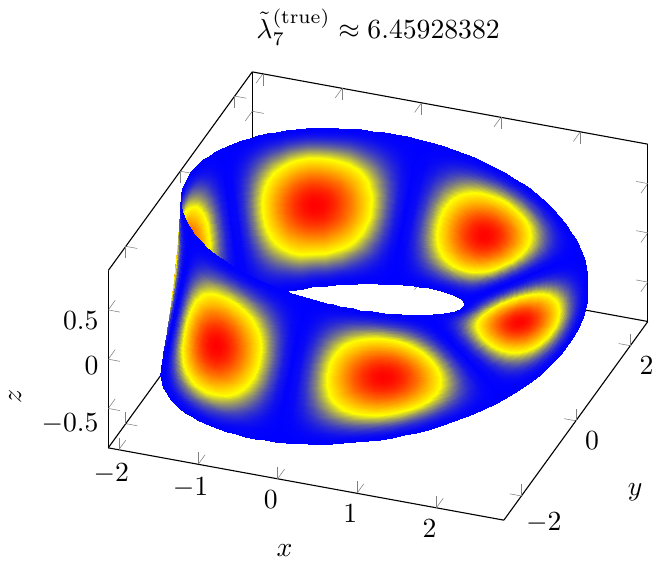}
    \hfill
    \includegraphics[scale=0.7]{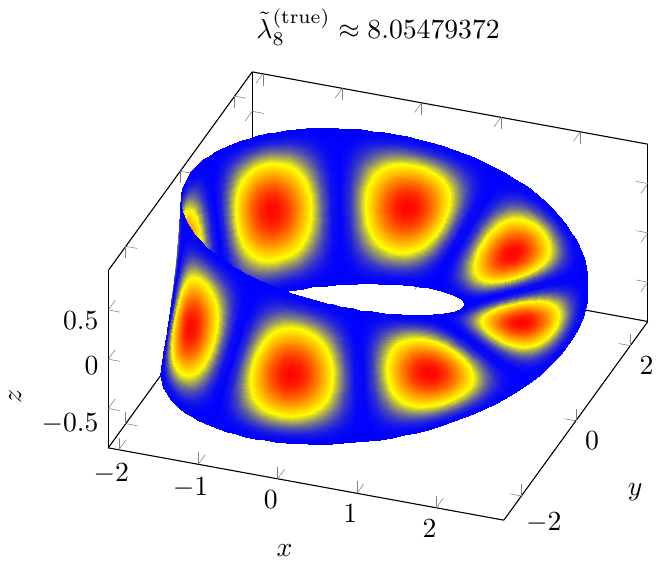}
    \hfill
    \includegraphics[scale=0.7]{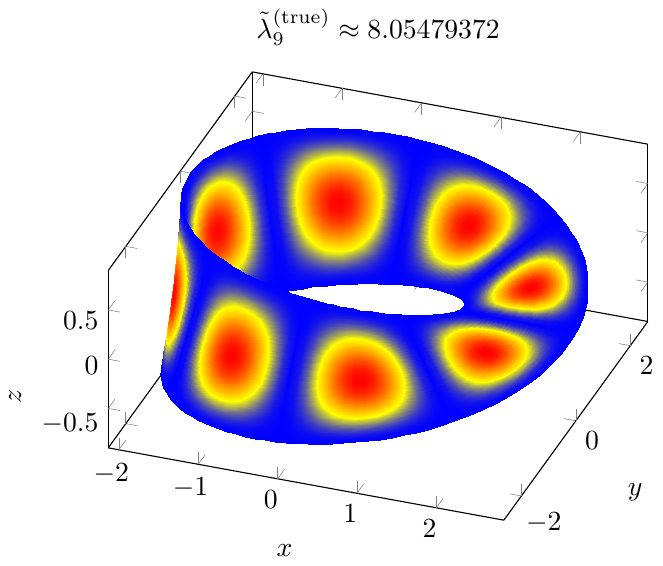}
    \par\medskip
    \includegraphics[scale=0.7]{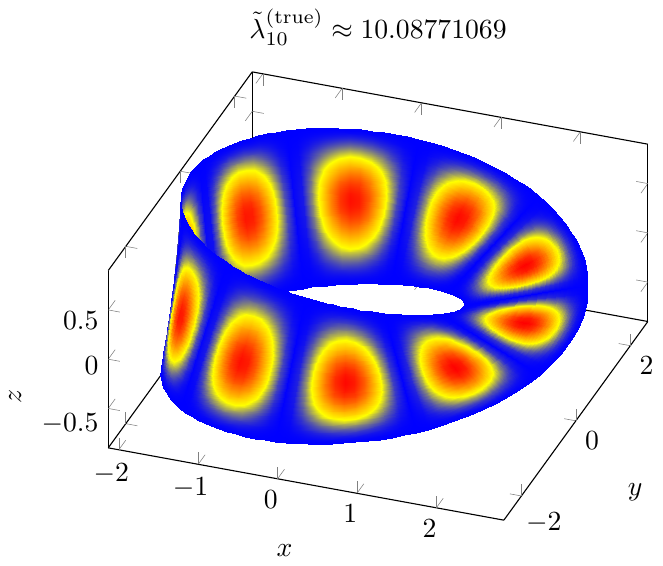}
  \par}
  \caption{
    Numerical approximations $\tilde{f}_k$ of eigenvalues and eigenfunctions of the operator~$L$ for $k=1,2,\ldots,10$ plotted onto the original Möbius strip.
    Parameters of the numerical computation are $a = 0.75$, $R = 13.2/(2\pi)$, and $N = 82$.
    Color coding is the same as in Figure~\ref{fig.eigenvectors.1}.
  }
  \label{fig.eigenvectors.3d.1}
\end{figure}
\begin{figure}
  \ContinuedFloat
  {\centering
    \includegraphics[scale=0.7]{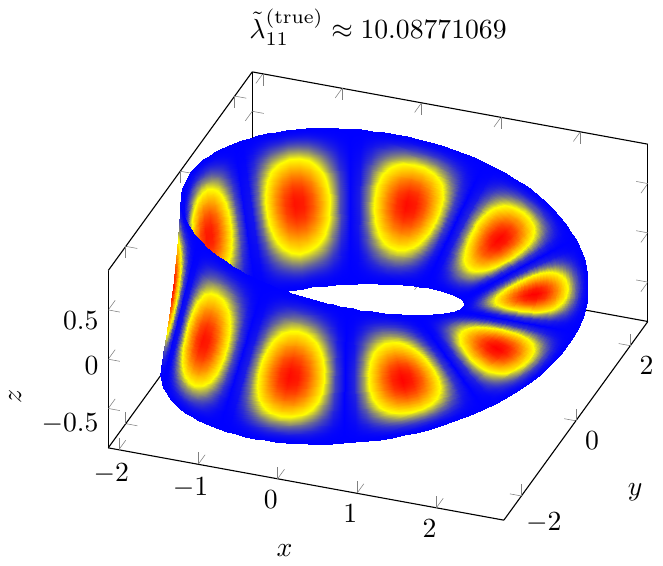}
    \hfill
    \includegraphics[scale=0.7]{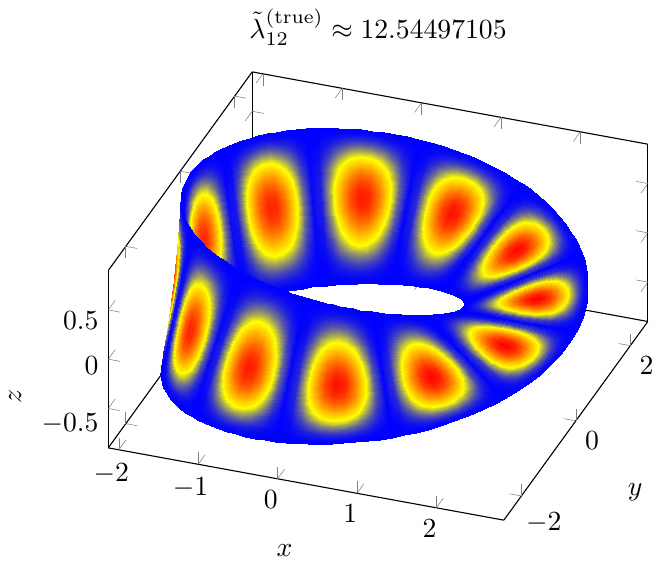}
    \hfill
    \includegraphics[scale=0.7]{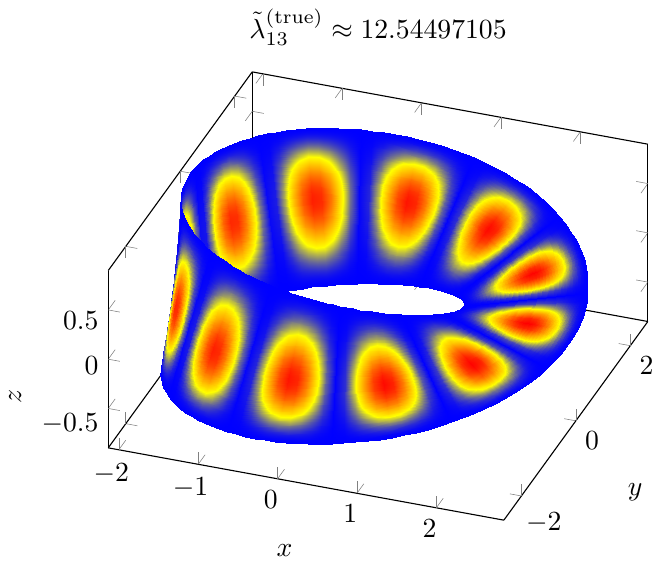}
    \par\medskip
    \includegraphics[scale=0.7]{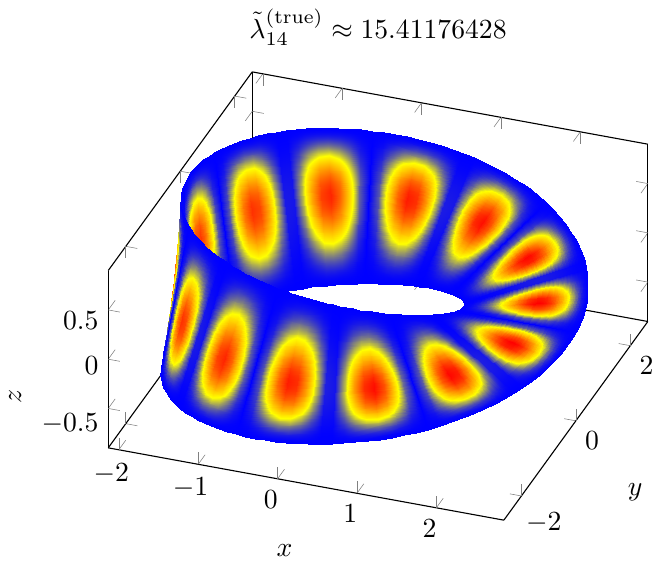}
    \hfill
    \includegraphics[scale=0.7]{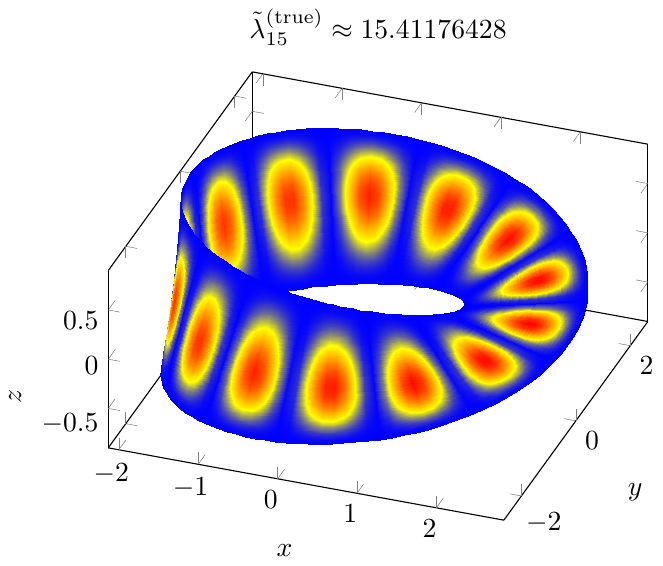}
    \hfill
    \includegraphics[scale=0.7]{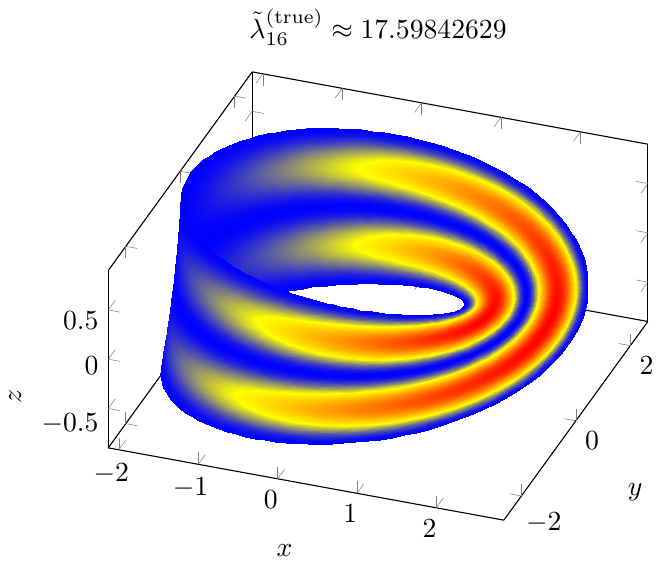}
    \par\medskip
    \includegraphics[scale=0.7]{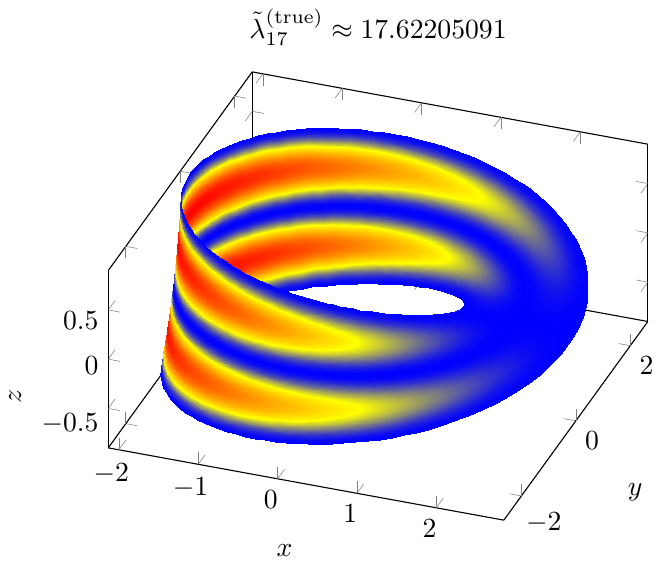}
    \hfill
    \includegraphics[scale=0.7]{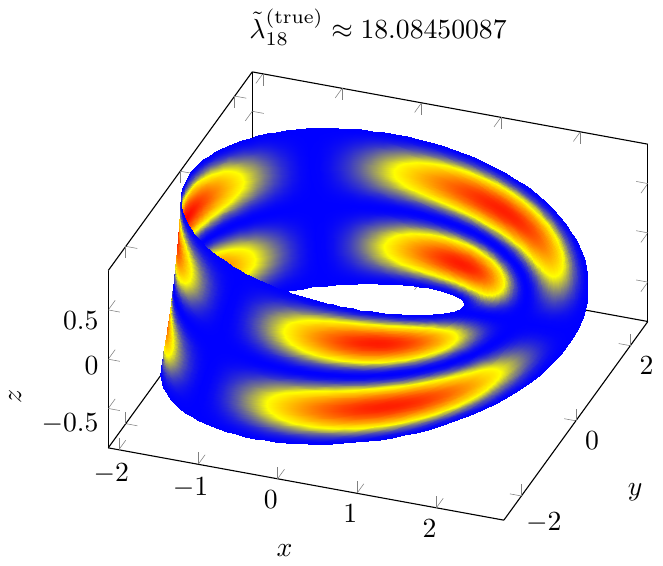}
    \hfill
    \includegraphics[scale=0.7]{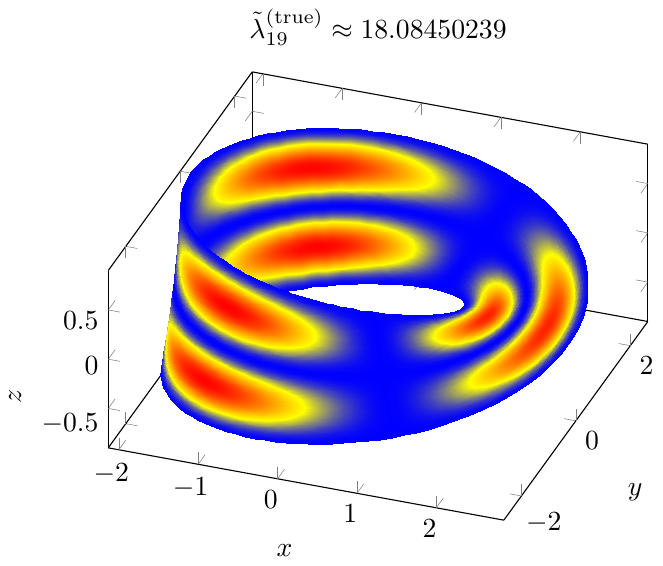}
    \par\medskip
    \includegraphics[scale=0.7]{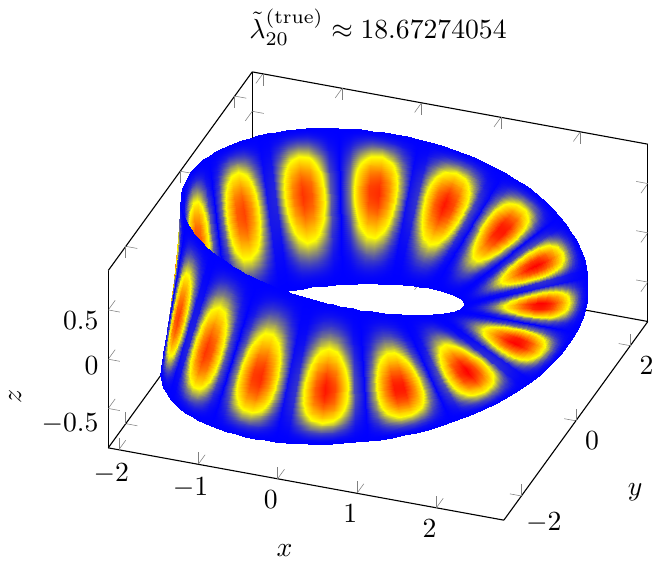}
  \par}
  \caption{
    Numerical approximations $\tilde{f}_k$ of eigenvalues and eigenfunctions of the operator~$L$ for $k=11,12,\ldots,20$ plotted onto the original Möbius strip.
    Parameters of the numerical computation are $a = 0.75$, $R = 13.2/(2\pi)$, and $N = 82$.
    Color coding is the same as in Figure~\ref{fig.eigenvectors.1}.
  }
  \label{fig.eigenvectors.3d.2}
\end{figure}
\begin{figure}
  {\centering
    \includegraphics{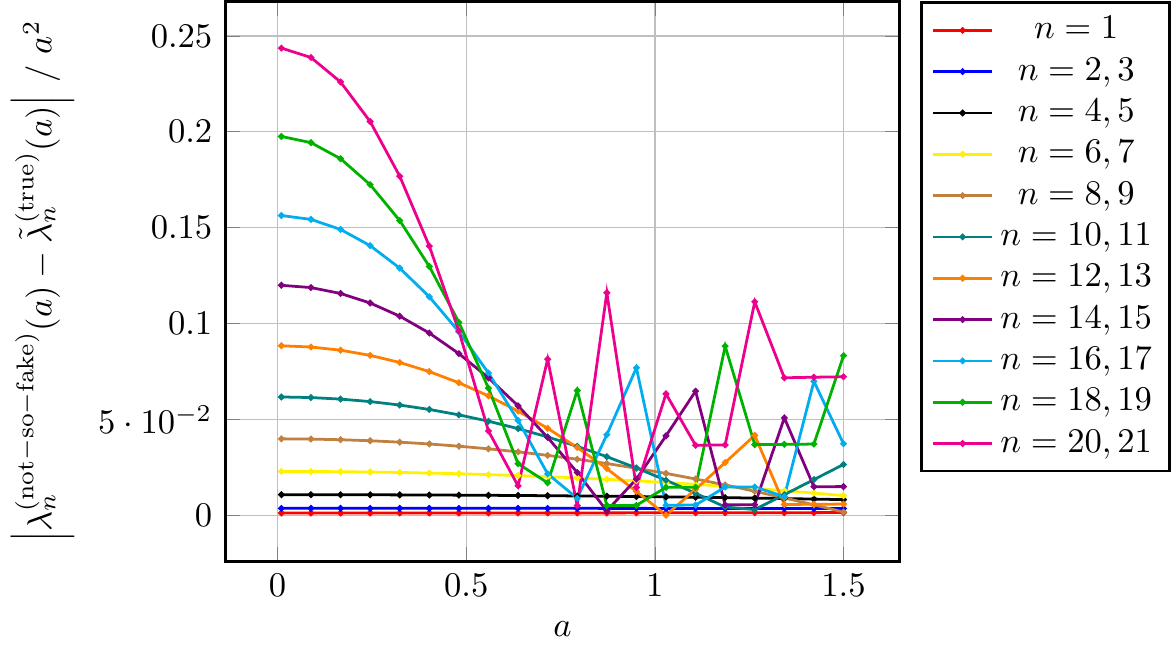}
  \par}
  \caption{
    Ratio of the difference of the $n$th eigenvalue of the not-so-fake model 
    $\lambda_n^{(\textrm{not-so-fake})}(a)$ and numerical approximation of the $n$th eigenvalue of the full model $\tilde{\lambda}_n^{(\mathrm{full})}(a)$ and $a^2$, $n=1,\ldots,20$.
    Values of parameters are $R = 18/(2\pi)$, $a$ ranges from $0.01$ to $1.5$, and $N=72$.
    Each curve, except the one for $n=1$, in fact represents two of these ratios which are indistinguishable in this plot resolution. 
  }
  \label{fig.diff.ratio}
\end{figure}
\begin{figure}
  {\centering
    \includegraphics{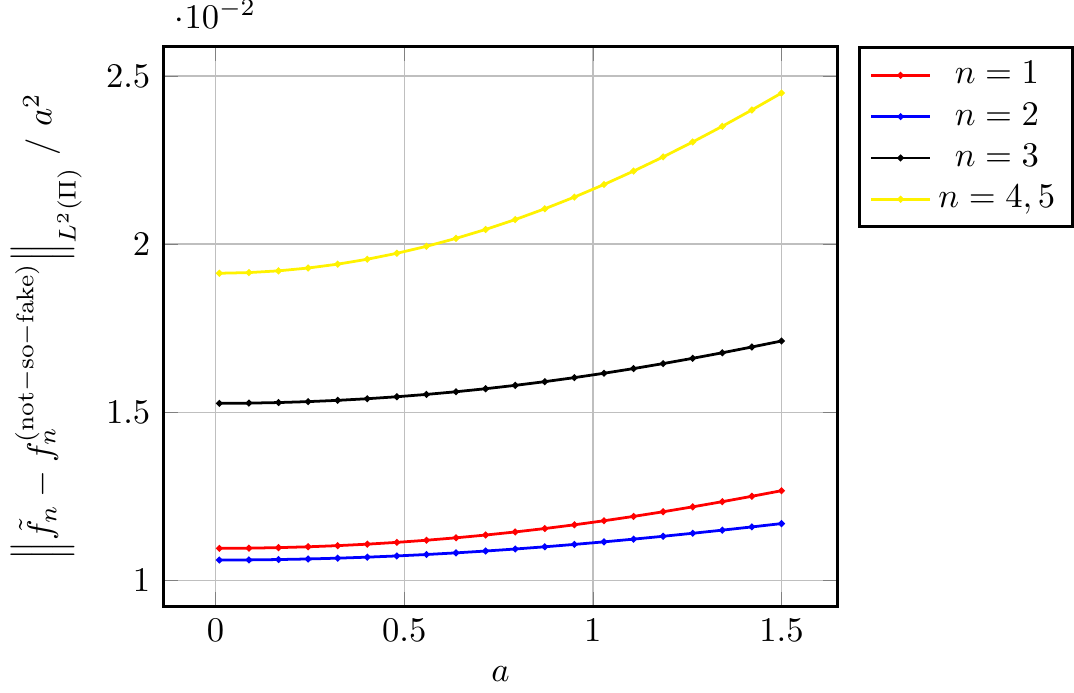}
  \par}
  \caption{
    Ratio of the norm of the difference of the $n$th eigenvector of the not-so-fake model 
    $f_n^{(\textrm{not-so-fake})}$ and numerical approximation of the $n$th eigenvector of the full model $\tilde{f}_n$ and $a^2$, $n=1,\ldots,5$.
    Values of parameters are $R = 18/(2\pi)$, $a$ ranges from $0.01$ to $1.5$, and $N=72$.
  }
  \label{fig.eigenvectors.convergence}
\end{figure}

\end{document}